\documentclass[envcountsame,envcountsect,runningheads]{llncs} 

\newcommand{\ShortVersion}[1]{} 
\newcommand{\LongVersion}[1]{#1} 
\newcommand{\VSpace}[1]{} 

\usepackage{latexsym} 
\usepackage{xspace} 
\usepackage{amssymb} 
\usepackage{amsmath}
\usepackage{stmaryrd} 
\usepackage{url} 
\usepackage{oldgerm} 
\input{xy} 
\xyoption{all} 
\usepackage[ruled,vlined,linesnumbered]{algorithm2e}

\newcommand{\mand}{\sqcap}
\newcommand{\mor}{\sqcup}
\newcommand{\V}{\forall}
\newcommand{\E}{\exists}
\def\tuple#1{\langle#1\rangle} 

\newcommand{\ovl}[1]{\overline{#1}}
\newcommand{\fracc}[2]{\displaystyle{\frac{\;#1\;}{\;#2\;}}}

\newcommand{\EXPTIME}{\textsc{Exp\-Time}\xspace}
\newcommand{\NEXPTIME}{\textsc{NExp\-Time}\xspace}

\newcommand{\mT}{\mathcal{T}} 
\newcommand{\mR}{\mathcal{R}} 
\newcommand{\mA}{\mathcal{A}} 
\newcommand{\mI}{\mathcal{I}} 
\newcommand{\mC}{\mathcal{C}} 
\newcommand{\mE}{\mathcal{E}} 

\def\trojkat{\mbox{{\scriptsize$\!\vartriangleleft$}}}

\def\defeq{=}

\def\eqref#1{(\ref{#1})}

\def\trojkat{\mbox{{\scriptsize$\!\vartriangleleft$}}}
\newcommand{\myEnd}{\mbox{}\hfill\trojkat}

\let\oldmarginpar\marginpar
\renewcommand\marginpar[1]{\oldmarginpar{\center{\footnotesize{\em #1}}}}

\newcommand{\comment}[1]{}

\newcommand{\CReg}{REG$^c$\xspace}

\newcommand{\SH}{$\mathcal{SH}$\xspace}
\newcommand{\SHI}{$\mathcal{SHI}$\xspace}
\newcommand{\SROIQ}{$\mathcal{SROIQ}$\xspace}

\newcommand{\ALC}{$\mathcal{ALC}$\xspace} 
\newcommand{\SHIQ}{$\mathcal{SHIQ}$\xspace} 
\newcommand{\ALCI}{$\mathcal{ALCI}$\xspace}

\newcommand{\CSHI}{$C_\mathcal{SHI}$\xspace}

\newcommand{\CN}{\mathbf{C}} 
 
\newcommand{\RN}{\mathbf{R}} 
\newcommand{\IN}{\mathbf{I}} 

\newcommand{\rBot}{(\bot)}
\newcommand{\rAnd}{(\mand)}
\newcommand{\rOr}{(\mor)}
\newcommand{\rH}{(H)}

\newcommand{\rTrans}{(\E)}

\newcommand{\rConv}{(conv)}
\newcommand{\rFormingState}{(forming\textrm{-}state)}

\newcommand{\rBotP}{(\bot')}
\newcommand{\rAndP}{(\mand')}
\newcommand{\rOrP}{(\mor')}
\newcommand{\rHP}{(H')}
\newcommand{\rV}{(\V^{'})}

\newcommand{\rTransP}{(\E')}

\newcommand{\closure}{\mathit{closure}}

\newcommand{\Type}{\mathit{Type}}
\newcommand{\SType}{\mathit{SType}}
\newcommand{\Status}{\mathit{Status}}
\newcommand{\Label}{\mathit{Label}}
\newcommand{\CELabel}{\mathit{CELabel}}
\newcommand{\StatePred}{\mathit{StatePred}}
\newcommand{\AfterTransPred}{\mathit{ATPred}}
\newcommand{\RFormulas}{\mathit{RFmls}}
\newcommand{\DFormulas}{\mathit{DFmls}}
\newcommand{\DSTimeStamp}{\mathit{DSTimeStamp}}
\newcommand{\ETimeStamp}{\mathit{ETimeStamp}}

\newcommand{\FmlsRC}{\mathit{FmlsRC}}
\newcommand{\AltFmlSetsSC}{\mathit{AltFmlSetsSC}}

\newcommand{\AltFmlSetsSCP}{\mathit{AltFmlSetsSCP}}
\newcommand{\ConvMethod}{\mathit{ConvMethod}}

\newcommand{\State}{\mathsf{state}}
\newcommand{\NonState}{\mathsf{non\textrm{-}state}}
\newcommand{\Complex}{\mathsf{complex}}
\newcommand{\Simple}{\mathsf{simple}}
\newcommand{\AndNode}{\mathsf{and\textrm{-}node}}
\newcommand{\OrNode}{\mathsf{or\textrm{-}node}}

\newcommand{\Sat}{\mathsf{sat}}
\newcommand{\Unsat}{\mathsf{unsat}}
\newcommand{\Incomplete}{\mathsf{incomplete}}
\newcommand{\Expanded}{\mathsf{expanded}}
\newcommand{\Unexpanded}{\mathsf{unexpanded}}
\newcommand{\Null}{\mathsf{null}}

\SetKwInput{Input}{Input}
\SetKwInput{Output}{Output}
\SetKwInput{GlobalData}{Global data}
\SetKwInput{Purpose}{Purpose}
\SetKw{Call}{call}
\SetKw{Set}{set}
\SetKw{Exit}{exit}
\SetKw{Break}{break}
\SetKw{Continue}{continue}

\SetKwFunction{NewSucc}{NewSucc} 
\SetKwFunction{FindProxy}{FindProxy} 
\SetKwFunction{ConToSucc}{ConToSucc} 
\SetKwFunction{Trans}{Trans}
\SetKwFunction{SRTR}{SRTR}
\SetKwFunction{Apply}{Apply}
\SetKwFunction{ApplyTransRule}{ApplyTransRule}
\SetKwFunction{ApplyConvRule}{ApplyConvRule}
\SetKwFunction{Tableau}{Tableau}
\SetKwFunction{Expand}{Expand}
\SetKwFunction{UpdateStatus}{UpdateStatus}
\SetKwFunction{UpdateUnsatNodes}{UpdateUnsatNodes}
\SetKwFunction{PropagateStatus}{PropagateStatus}
\SetKwFunction{Kind}{Kind}
\SetKwFunction{AfterTrans}{AfterTrans}
\SetKwFunction{FullLabel}{FullLabel}
\SetKwFunction{BeforeFormingState}{BeforeFormingState}
\SetKwFunction{TUnsat}{TUnsat}
\SetKwFunction{TSat}{TSat}
\SetKwFunction{ToExpand}{ToExpand}
\SetKwFunction{AFormulas}{AFmls}

\SetKwFunction{NDFormulas}{NDFmls}
\SetKwFunction{Cnj}{Cnj}
\SetKwFunction{NegCnj}{NegCnj}
\SetKwFunction{NegAll}{NegAll}

\newcommand{\Ext}{\mathit{ext}}

\ShortVersion{\title{A Cut-Free ExpTime Tableau Decision Procedure for the Description Logic SHI}}

\LongVersion{\title{Cut-Free ExpTime Tableaux\\ for Checking Satisfiability of a Knowledge Base\\ in the Description Logic SHI}}

\titlerunning{Cut-Free \EXPTIME\ Tableaux for the Description Logic {\SHI}}

\author{Linh Anh Nguyen}
\institute{
Institute of Informatics, University of Warsaw\\
Banacha 2, 02-097 Warsaw, Poland\\
\email{nguyen@mimuw.edu.pl}
}

\authorrunning{L.A. Nguyen}	

\tocauthor{Linh Anh Nguyen (University of Warsaw)}

\begin{document}
\maketitle \sloppy

\begin{abstract}
We give the first cut-free \EXPTIME (optimal) tableau decision procedure for checking satisfiability of a knowledge base in the description logic \SHI, which extends the description logic \ALC with transitive roles, inverse roles and role hierarchies.
\LongVersion{

\medskip
\noindent\textbf{Keywords}: description logics, automated reasoning, tableaux, global caching.
}
\end{abstract}


\section{Introduction}

Ontologies provide a shared understanding of the domain for different applications that want to communicate to each other. They are useful for several important areas like knowledge representation, software integration and Web applications. Web Ontology Language (OWL) is a layer of the Semantic Web architecture, built on the top of XML and RDF. Together with rule languages it serves as a main knowledge representation formalism for the Semantic Web. The logical foundation of OWL is based on description logics (DLs). Some of the most well-known DLs, in the increasing order of expressiveness, are \ALC, \SH, \SHI, \SHIQ and \SROIQ~\cite{dlbook,HorrocksKS06}.

Description logics represent the domain of interest in terms of concepts, individuals, and roles. A concept is interpreted as a set of individuals, while a role is interpreted as a binary relation among individuals. A knowledge base in a DL consists of axioms about roles (grouped into an RBox), terminology axioms (grouped into a TBox), and assertions about individuals (grouped into an ABox). One of the basic inference problems in DLs, which we denote by $Sat$, is to check satisfiability of a knowledge base. Other inference problems in DLs are usually reducible to this problem. For example, the problem of checking consistency of a concept w.r.t.\ an RBox and a TBox (further denoted by $Cons$) is linearly reducible to $Sat$.

In this paper we study automated reasoning in the DL \SHI, which extends the DL \ALC with transitive roles, inverse roles and role hierarchies. The aim is to develop an efficient tableau decision procedure for the $Sat$ problem in \SHI. It should be complexity-optimal (\EXPTIME), cut-free, and extendable with useful optimizations. Tableau methods have widely been used for automated reasoning in modal and description logics~\cite{DAgostino1999} since they are natural and allow many optimizations. As \SHI is a sublogic of \SROIQ and \CReg (regular grammar logic with converse), one can use the tableau decision procedures of \SROIQ~\cite{HorrocksKS06} and \CReg~\cite{NguyenSzalas-SL} for the $Sat$ problem in \SHI. However, the first procedure has suboptimal complexity (\NEXPTIME when restricted to \SHI), and the second one uses analytic cuts. 

The tableau decision procedure given in~\cite{HorrockSattler99} for the $Cons$ problem in \SHI has \NEXPTIME complexity.  
In~\cite{GoreNguyenTab07} together with Gor{\'e} we gave the first \EXPTIME tableau decision procedure for the $Cons$ problem in \SHI, which uses analytic cuts to deal with inverse roles. In~\cite{NguyenS10TCCI} together with Sza{\l}as we gave the first direct \EXPTIME tableau decision procedure for the $Sat$ problem in the DL \SH. In~\cite{Nguyen-ALCI} we gave the first cut-free \EXPTIME tableau decision procedure for the $Sat$ problem in the DL \ALCI.  

In this paper, by extending the methods of~\cite{GoreNguyenTab07,NguyenS10TCCI,Nguyen-ALCI}, we give the first cut-free \EXPTIME (optimal) tableau decision procedure for the $Sat$ problem in the DL \SHI. 
We use global state caching~\cite{GoreW09,GoreW10,Nguyen-ALCI}, the technique of~\cite{Nguyen-ALCI} for dealing with inverse roles, the technique of~\cite{GoreNguyenTab07,NguyenS10TCCI} for dealing with transitive roles and hierarchies of roles, and the techniques of~\cite{NguyenS10TCCI,NguyenS10FI,NguyenSzalas-SL,Nguyen-ALCI} for dealing with ABoxes. 

The rest of this paper is structured as follows: In Section~\ref{section: prel} we recall the notation and semantics of \SHI. In Section~\ref{section: calculus} we\LongVersion{ present}\ShortVersion{ describe} our tableau decision procedure for the $Sat$ problem in \SHI.\LongVersion{ In Section~\ref{section: proofs} we give proofs for the correctness of our procedure and analyze its complexity.} Section~\ref{section: conc} concludes this work.\ShortVersion{ Due to the lack of space, pseudocode of our decision procedure and proofs of our results are presented only in the long version~\cite{nSHI-long} of the current paper.}


\section{Notation and Semantics of \SHI}
\label{section: prel}

Our language uses a finite set $\CN$ of {\em concept names}, a finite set $\RN$ of role names, and a finite set $\IN$ of individual names.  
We use letters like $A$ and $B$ for {\em concept names}, $r$ and $s$ for {\em role names}, and  $a$ and $b$ for {\em individual names}. We refer to $A$ and $B$ also as {\em atomic concepts}, and to $a$ and $b$ as {\em individuals}. 

For $r \in \RN$, let $r^-$ be a new symbol, called the {\em inverse} of $r$. 
Let $\RN^{-} \defeq \{r^{-} \mid r \in \RN\}$ be the set of {\em inverse roles}. For $r \in \RN$, define $(r^-)^- \defeq r$. A {\em role} is any member of $\RN \cup \RN^{-}$. We use letters like $R$ and $S$ to denote roles.  

An (\SHI) {\em RBox} $\mR$ is a finite set of role axioms of the form $R \sqsubseteq S$ or $R \circ R \sqsubseteq R$. 
By $\Ext(\mR)$ we denote the least extension of $\mR$ such that:
\begin{itemize}
\item $R \sqsubseteq R \in \Ext(\mR)$ for any role $R$
\item if $R \sqsubseteq S \in \Ext(\mR)$ then $R^- \sqsubseteq S^- \in \Ext(\mR)$
\item if $R \circ R \sqsubseteq R \in \Ext(\mR)$ then $R^- \circ R^- \sqsubseteq R^- \in \Ext(\mR)$
\item if $R \sqsubseteq S \in \Ext(\mR)$ and $S \sqsubseteq T \in \Ext(\mR)$ then $R \sqsubseteq T \in \Ext(\mR)$.
\end{itemize}

By $R \sqsubseteq_\mR S$ we mean $R \sqsubseteq S \in \Ext(R)$. If $R \sqsubseteq_\mR S$ then $R$ is a~{\em subrole} of $S$ w.r.t.~$\mR$.
If $R \circ R \sqsubseteq R \in \Ext(\mR)$ then $R$ is a~{\em transitive role} w.r.t.~$\mR$. 

{\em Concepts} in \SHI\ are formed using the following BNF
grammar:
\[
C, D ::=
        \top
        \mid \bot
        \mid A
        \mid \lnot C
        \mid C \mand D
        \mid C \mor D
        \mid \V R.C
        \mid \E R.C
\]

We use letters like $C$ and $D$ to denote arbitrary concepts.

A {\em TBox} is a~finite set of axioms of the form $C
\sqsubseteq D$ or $C \doteq D$.
An {\em ABox} is a~finite set of {\em assertions} of the form
$a\!:\!C$ ({\em concept assertion}) or $R(a, b)$ ({\em role
assertion}).
A {\em knowledge base} in \SHI\ is a tuple $(\mR,\mT,\mA)$, where 
$\mR$ is an RBox, $\mT$ is a~TBox and $\mA$ is an ABox.

A {\em formula} is defined to be either a concept or an ABox assertion. 
We use letters like $\varphi$, $\psi$ and $\xi$ to denote formulas, and letters like $X$, $Y$ and $\Gamma$ to denote sets of formulas.

An {\em interpretation} $\mI = \langle \Delta^\mI, \cdot^\mI
\rangle$ consists of a~non-empty set $\Delta^\mI$, called the {\em
domain} of $\mI$, and a~function $\cdot^\mI$, called the {\em
interpretation function} of $\mI$, that maps every concept name
$A$ to a~subset $A^\mI$ of $\Delta^\mI$, maps every role name
$r$ to a~binary relation $r^\mI$ on $\Delta^\mI$, and maps
every individual name $a$ to an element $a^\mI \in \Delta^\mI$.
The interpretation function is extended to inverse roles and complex concepts as
follows:
\[
\begin{array}{c}
(r^-)^\mI = \{ \tuple{x,y} \mid \tuple{y,x} \in r^\mI\} 
\quad\quad\quad\quad
	\top^\mI = \Delta^\mI
\quad\quad\quad\quad
	\bot^\mI = \emptyset
\\[1.0ex]
(\lnot C)^\mI = \Delta^\mI \setminus C^\mI
\quad\quad 
	(C \mand D)^\mI = C^\mI \cap D^\mI
\quad\quad 
	(C \mor D)^\mI = C^\mI \cup D^\mI
\\[1.0ex]
(\V R.C)^\mI =
        \big\{ x \in \Delta^\mI \mid \V y\big[ (x,y) \in R^\mI
                       \textrm{ implies } y \in C^\mI\big]\big\}
\\[1.0ex]
(\E R.C)^\mI =
        \big\{ x \in \Delta^\mI \mid \E y\big[ (x,y) \in R^\mI
                 \textrm{ and }\ y \in C^\mI\big]\big\}
\end{array}
\]

Note that $(r^-)^\mI = (r^\mI)^{-1}$ and this is compatible with $(r^-)^- = r$.

For a set $\Gamma$ of concepts, define $\Gamma^\mI = \{x \in \Delta^\mI \mid x \in C^\mI \textrm{ for all } C \in \Gamma\}$.

The relational composition of binary relations $\mathrm{R}_1$, $\mathrm{R}_2$ is denoted by $\mathrm{R}_1 \circ \mathrm{R}_2$.

An interpretation $\mI$ is a~{\em model of an RBox} $\mR$ if for every axiom $R \sqsubseteq S$ (resp.\ $R \circ R \sqsubseteq R$) of $\mR$, we have that $R^\mI \subseteq S^\mI$ (resp.\ $R^\mI \circ R^\mI \subseteq R^\mI$). 
Note that if $\mI$ is a model of $\mR$ then it is also a model of $\Ext(\mR)$.

An interpretation $\mI$ is a~{\em model of a~TBox} $\mT$ if for
every axiom $C \sqsubseteq D$ (resp.\ $C \doteq D$) of $\mT$,
we have that $C^\mI \subseteq D^\mI$ (resp.\ $C^\mI = D^\mI$).

An interpretation $\mI$ is a~{\em model of an ABox} $\mA$ if
for every assertion $a\!:\!C$ (resp.\ $R(a,b)$) of $\mA$, we
have that $a^\mI \in C^\mI$ (resp.\ $(a^\mI,b^\mI) \in R^\mI$).

An interpretation $\mI$ is a~{\em model of a~knowledge base}
$(\mR,\mT,\mA)$ if $\mI$ is a~model of all $\mR$, $\mT$ and $\mA$.
A knowledge base $(\mR,\mT,\mA)$ is {\em satisfiable} if it has
a~model.

An interpretation $\mI$ {\em satisfies} a concept $C$ (resp.\ a set $X$ of concepts) if $C^\mI \neq \emptyset$ (resp.\ $X^\mI \neq \emptyset$).
A set $X$ of concepts is {\em satisfiable w.r.t.\ an RBox $\mR$ and a TBox $\mT$} if there exists a model of $\mR$ and $\mT$ that satisfies $X$. For $X = Y \cup \mA$, where $Y$ is a set of concepts and $\mA$ is an ABox, we say that $X$ is {\em satisfiable w.r.t.\ an RBox $\mR$ and a TBox $\mT$} if there exists a~model of $(\mR,\mT,\mA)$ that satisfies~$X$.


\section{A Tableau Decision Procedure for \SHI}
\label{section: calculus}

We assume that concepts and ABox assertions are represented in
negation normal form (NNF), where $\lnot$ occurs only directly
before atomic concepts.\footnote{Every formula can be
transformed to an equivalent formula in NNF.} 
We use $\overline{C}$ to denote the NNF of $\lnot C$, and for $\varphi = a\!:\!C$, we use $\ovl{\varphi}$ to denote $a\!:\!\ovl{C}$.
For simplicity, we treat axioms
of $\mT$ as concepts representing global assumptions: an axiom
$C \sqsubseteq D$ is treated as $\overline{C} \mor D$, while an
axiom $C \doteq D$ is treated as $(\overline{C} \mor D) \mand
(\overline{D} \mor C)$. That is, we assume that $\mT$ consists
of concepts in NNF. Thus, an interpretation $\mI$ is a~model of
$\mT$ iff $\mI$ validates every concept $C \in \mT$. 
As this way of handling the TBox is not efficient in practice, the absorption technique like the one discussed in~\cite{NguyenSzalas09ICCCI,NguyenS10TCCI} can be used to improve the performance of our algorithm.

From now on, let $(\mR,\mT,\mA)$ be a knowledge base in NNF of the logic \SHI, with $\mA \neq \emptyset$.$\,$\footnote{If $\mA$ is empty, we can add $a\!:\!\top$ to it, where $a$ is a special individual.} 
In this section we present a tableau calculus for checking satisfiability of $(\mR,\mT,\mA)$.

For a set $X$ of concepts and a set $Y$ of ABox assertions, we define:
\begin{eqnarray*}
\SRTR_\mR(R,S) & = & (R \sqsubseteq_\mR S \land S \circ S \sqsubseteq S \in \Ext(\mR)) \\
\Trans_\mR(X,R) & = & \{D \mid \V R.D \in X\} \cup \{\V S.D \in X \mid \SRTR_\mR(R,S) \}\\
\Trans_\mR(X,R,a) & = & \{a\!:\!D \mid \V R.D \in X\} \cup \{a\!:\!\V S.D \mid \V S.D \in X \land \SRTR_\mR(R,S)\}\\
\Trans_\mR(Y,a,R) & = & \{D \mid a\!:\!\V R.D \in Y\} \cup \{\V S.D \mid a\!:\!\V S.D \in Y \land \SRTR_\mR(R,S)\}\\
\Trans_\mR(Y,a,R,b) & = & \{b\!:\!D \mid a\!:\!\V R.D \in Y\}\ \cup\\
	& & \{b\!:\!\V S.D \mid a\!:\!\V S.D \in Y \land \SRTR_\mR(R,S)\}
\end{eqnarray*}
We call $\Trans_\mR(X,R)$ the {\em transfer of $X$ through $R$ w.r.t.\ $\mR$}, call $\Trans_\mR(X,R,a)$ the {\em transfer of $X$ through $R$ to $a$ w.r.t.\ $\mR$}, call $\Trans_\mR(Y,a,R)$ the {\em transfer of $Y$ starting from $a$ through $R$ w.r.t.\ $\mR$}, and call $\Trans_\mR(Y,a,R,b)$ the {\em transfer of $Y$ from $a$ to $b$ through $R$ w.r.t.\ $\mR$}.


In what follows we define tableaux as rooted ``and-or'' graphs. Such a graph is a tuple $G = (V,E,\nu)$, where $V$ is a set of nodes, $E \subseteq V \times V$ is a set of edges, $\nu \in V$ is the root, and each node $v \in V$ has a number of attributes. If there is an edge $(v,w) \in E$ then we call $v$ a {\em predecessor} of $w$, and call $w$ a {\em successor} of~$v$. The set of all attributes of $v$ is called the {\em contents of $v$}. Attributes of tableau nodes are:
\begin{itemize}
\item $\Type(v) \in \{\State, \NonState\}$. If $\Type(v) = \State$ then we call $v$ a {\em state}, else we call $v$ a {\em non-state} (or an {\em internal} node). If $\Type(v) = \State$ and $(v,w) \in E$ then $\Type(w) = \NonState$.

\item $\SType(v) \in \{\Complex, \Simple\}$ is called the subtype of $v$. If $\SType(v) = \Complex$ then we call $v$ a {\em complex node}, else we call $v$ a {\em simple node}. The graph never contains edges from a simple node to a complex node. If $(v,w)$ is an edge from a complex node $v$ to a simple node $w$ then $\Type(v) = \State$ and $\Type(w) = \NonState$. The root of the graph is a complex node. 

\item $\Status(v) \in \{\Unexpanded$, $\Expanded$, $\Incomplete$, $\Unsat$, $\Sat\}$.

\item $\Label(v)$ is a finite set of formulas, called the label of $v$. The label of a complex node consists of ABox assertions, while the label of a simple node consists of concepts.

\item $\RFormulas(v)$ is a finite set of formulas, called the set of reduced formulas of~$v$.

\item $\DFormulas(v)$ is a finite set of formulas, called the set of disallowed formulas of~$v$.

\item $\StatePred(v) \in V \cup \{\Null\}$ is called the state-predecessor of $v$. It is available only when $\Type(v) = \NonState$. 
If $v$ is a non-state and $G$ has no paths connecting a state to $v$ then $\StatePred(v) = \Null$. Otherwise, $G$ has exactly one state $u$ that is connected to $v$ via a path not containing any other states. In that case, $\StatePred(v) = u$.

\item $\AfterTransPred(v) \in V$ is called the after-transition-predecessor of~$v$. It is available only when $\Type(v) = \NonState$. If $v$ is a non-state and $v_0 = \StatePred(v)$ ($\neq \Null$) then there is exactly one successor $v_1$ of $v_0$ such that every path connecting $v_0$ to $v$ must go through $v_1$, and we have that $\AfterTransPred(v) = v_1$. 
We define $\AfterTrans(v) = (\AfterTransPred(v) = v)$. 
If $\AfterTrans(v)$ holds then either $v$ has no predecessors (i.e.\ it is the root of the graph) or it has exactly one predecessor $u$ and $u$ is a state. 

\item $\CELabel(v)$ is a formula called the coming edge label of $v$. It is available only when $v$ is a successor of a state $u$ (and $\Type(v) = \NonState$). In that case, we have $u = \StatePred(v)$, $\AfterTrans(v)$ holds, $\CELabel(v) \in \Label(u)$, and
  \begin{itemize}
  \item if $\SType(u) = \Simple$ then\\ \mbox{\hspace{1.9em}}$\CELabel(v)$ is of the form $\E R.C$ and $C \in \Label(v)$
  \item else $\CELabel(v)$ is of the form $a\!:\!\E R.C$ and $C \in \Label(v)$.
  \end{itemize}
Informally, $v$ was created from $u$ to realize the formula $\CELabel(v)$ at~$u$.

\item $\ConvMethod(v) \in \{0,1\}$ is called the converse method of $v$. It is available only when $\Type(v) = \State$.

\item $\FmlsRC(v)$ is a set of formulas, called the set of formulas required by converse for $v$. It is available only when $\Type(v) = \State$ and will be used only when $\ConvMethod(v) = 0$. 

\item $\AltFmlSetsSC(v)$ is a set of sets of formulas, called the set of alternative sets of formulas suggested by converse for $v$. It is available only when $\Type(v) = \State$ and will be used only when $\ConvMethod(v) = 1$. 

\item $\AltFmlSetsSCP(v)$ is a set of sets of formulas, called the set of alternative sets of formulas suggested by converse for the predecessor of $v$. It is available only when $v$ has a predecessor being a state and will be used only when $\ConvMethod(v) = 1$. 
\end{itemize}

We define 
\VSpace{-1.5ex}
\begin{eqnarray*}
\AFormulas(v) & = & \Label(v) \cup \RFormulas(v) \\
\NDFormulas(v) & = & \{\ovl{\varphi} \mid \varphi \in \DFormulas(v)\} \\
\FullLabel(v) & = & \AFormulas(v) \cup \NDFormulas(v) \\ 
\Kind(v) & = & 
  \left\{
  \begin{array}{l}
  \AndNode \textrm{ if } \Type(v) = \State\\
  \OrNode  \textrm{ if } \Type(v) = \NonState
  \end{array}
  \right.\\
\BeforeFormingState(v) & = & \textrm{$v$ has a successor which is a state}
\end{eqnarray*}

The sets $\AFormulas(v)$, $\NDFormulas(v)$, and $\FullLabel(v)$ are respectively called the available formulas of $v$, the negations of the formulas disallowed at $v$, and the full label of~$v$. 
In an ``and-or'' graph, states play the role of {\em ``and''-nodes}, while non-states play the role of {\em ``or''-nodes}. 

By the {\em local graph} of a state $v$ we mean the subgraph of $G$ consisting of all the path starting from $v$ and not containing any other states. 
Similarly, by the local graph of a non-state $v$ we mean the subgraph of $G$ consisting of all the path starting from $v$ and not containing any states.

We apply global state caching: if $v_1$ and $v_2$ are different states then $\Label(v_1) \neq \Label(v_2)$ or $\RFormulas(v_1) \neq \RFormulas(v_2)$ or $\DFormulas(v_1) \neq \DFormulas(v_2)$. 
If $v$ is a non-state such that $\AfterTrans(v)$ holds then we also apply global caching for the local graph of $v$: if $w_1$ and $w_2$ are different nodes of the local graph of $v$ then $\Label(w_1) \neq \Label(w_2)$ or $\RFormulas(w_1) \neq \RFormulas(w_2)$ or $\DFormulas(w_1) \neq \DFormulas(w_2)$. 


\begin{table}[t!]
$$
\begin{array}{|c|}
\hline
\ \\
\rAnd\, \fracc{X, C \mand D}{X, C, D} \quad\quad\quad
\rOr\, \fracc{X, C \mor D}{X, C \mid X, D} \quad\quad\quad
\rH\, \fracc{X, \V S.C}{X, \V S.C, \V R.C}\;\textrm{if $R \sqsubseteq_\mR S$}\\
\ \\
\hline
\ \\
\rTrans\; \fracc{X, \E R_1.C_1, \ldots, \E R_k.C_k}{C_1, X_1, \mT \;\&\; \ldots \;\&\; C_k, X_k, \mT}\;\; \textrm{if} 
  \left\{
  \begin{array}{l}
  \textrm{$X$ contains no concepts of the}\\
  \textrm{form $\E R.D$ and, for $1 \leq i \leq k$,}\\
  \textrm{$X_i = \Trans_\mR(X,R_i)$}
  \end{array}
  \right.\;\\
\ \\
\hline\hline
\ \\
\, \rAndP\, \fracc{X,\, a\!:\!(C \mand D)}{X,\, a\!:\!C,\, a\!:\!D} \;\;\ 
\rOrP\, \fracc{X,\, a\!:\!(C \mor D)}{X,\, a\!:\!C \mid X,\, a\!:\!D} \;\;\ 
\rHP\, \fracc{X, a\!:\!\V S.C}{X, a\!:\!\V S.C, a\!:\!\V R.C}\;\textrm{if $R \sqsubseteq_\mR S$}\,\\
\ \\
\hline
\ \\
\rV\; \fracc{X,\, R(a,b)}{X,\, R(a,b),\, \Trans(X,a,R,b),\, \Trans(X,b,R^-,a)} \\
\ \\
\hline
\ \\
\rTransP\; \fracc{X,\, a_1\!:\!\E R_1.C_1, \ldots,\, a_k\!:\!\E R_k.C_k}{C_1, X_1, \mT \;\&\; \ldots \;\&\; C_k, X_k, \mT}\;\; \textrm{if} 
  \left\{
  \begin{array}{l}
  \textrm{$X$ contains no assertions of the}\\
  \textrm{form $a\!:\!\E R.D$ and, for $1 \leq i \leq k$,}\\
  \textrm{$X_i = \Trans_\mR(X,a_i,R_i)$}
  \end{array}
  \right.\;\\
\ \\
\hline
\end{array}
$$
\caption{Some rules of the tableau calculus \CSHI \label{table: CSHI}}
\VSpace{-3.0ex}
\end{table}


\LongVersion{
\begin{figure*}[ht!]
\begin{function}[H]
\caption{NewSucc($v, type, sType, ceLabel, label, rFmls, dFmls$)\label{proc: NewSucc}}
\GlobalData{a rooted graph $(V,E,\nu)$.}
\Purpose{create a new successor for $v$.}
\LinesNumberedHidden
create a new node $w$,\ \ 
$V := V \cup \{w\}$,\ \ 
\lIf{$v \neq \Null$}{$E := E \cup \{(v,w)\}$}\;

$\Type(w) := type$,\ 
$\SType(w) := sType$,\ 
$\Status(w) := \Unexpanded$\; 
$\Label(w) := label$,\ 
$\RFormulas(w) := rFmls$,\ 
$\DFormulas(w) := dFmls$\; 

\uIf{$type = \NonState$}{
   \lIf{$v = \Null$ or $\Type(v) = \State$}{$\StatePred(w) := v$,\ \ 
	$\AfterTransPred(w) := w$\\
   }
   \lElse{$\StatePred(w) := \StatePred(v)$,\ \ 
	$\AfterTransPred(w) := \AfterTransPred(v)$}\;
   \lIf{$\Type(v) = \State$}{$\CELabel(w) := ceLabel$,\ \ 
	$\AltFmlSetsSCP(w) := \emptyset$}
}
\lElse{$\ConvMethod(w) := 0$,\ \ 
   $\FmlsRC(w) := \emptyset$,\ \ 
   $\AltFmlSetsSC(w) := \emptyset$
}\\

\Return{w}
\end{function}

\begin{function}[H]
\caption{FindProxy($type, sType, v_1, label, rFmls, dFmls$)\label{proc: FindProxy}}
\GlobalData{a rooted graph $(V,E,\nu)$.}
\LinesNumberedHidden

\lIf{$type = \State$}{$W := V$}
\lElse{$W := $ the nodes of the local graph of $v_1$}\;

\lIf{there exists $w \in W$ such that $\Type(w) = type$ and $\SType(w) = sType$ and $\Label(w) = label$ and $\RFormulas(w) = rFmls$ and $\DFormulas(w) = dFmls$}{\Return $w$\\}
\lElse{\Return $\Null$}
\end{function}

\begin{function}[H]
\caption{ConToSucc($v, type, sType, ceLabel, label, rFmls, dFmls$)\label{proc: ConToSucc}}
\GlobalData{a rooted graph $(V,E,\nu)$.}
\Purpose{connect $v$ to a successor, which is created if necessary.}

\lIf{$type = \State$}{$v_1 := \Null$}
\lElse{$v_1 := \AfterTransPred(v)$}

$w := \FindProxy(type, sType, v_1, label, rFmls, dFmls)$\;
\lIf{$w \neq \Null$}{$E := E \cup \{(v,w)\}$\\}
\lElse{$w := \NewSucc(v, type, sType, ceLabel, label, rFmls, dFmls)$}\;

\Return{w}
\end{function}

\begin{function}[H]
\caption{TUnsat($v$)\label{proc: TUnsat}}
\LinesNumberedHidden

\Return{($\bot \in \Label(v)$ or there exists $\{\varphi,\ovl{\varphi}\} \subseteq \Label(v)$)}
\end{function}

\begin{function}[H]
\caption{TSat($v$)\label{proc: TSat}}
\LinesNumberedHidden

\Return{($\Status(v) = \Unexpanded$ and no rule except $\rConv$ is applicable to $v$)}
\end{function}

\begin{function}[H]
\caption{ToExpand()\label{proc: ToExpand}}
\GlobalData{a rooted graph $(V,E,\nu)$.}
\lIf{there exists a node $v \in V$ with $\Status(v) = \Unexpanded$}{\Return $v$\\}
\lElse{\Return $\Null$}
\end{function}
\VSpace{-3ex}
\end{figure*}


\begin{figure*}[ht!]
\begin{procedure}[H]
\caption{Apply($\rho, v$)\label{proc: Apply}}
\GlobalData{a rooted graph $(V,E,\nu)$.}
\Input{a rule $\rho$ and a node $v \in V$ s.t. if $\rho \neq \rConv$ then $\Status(v) = \Unexpanded$ else $\Status(v) = \Expanded$ and $\BeforeFormingState(v)$ holds.}
\Purpose{applying the tableau rule $\rho$ to the node $v$.}

\BlankLine

  \uIf{$\rho = \rFormingState$}{
     $\ConToSucc(v,\State, \SType(v),\Null,\Label(v),\RFormulas(v), \DFormulas(v))$
  }
  \lElseIf{$\rho = \rConv$}{$\ApplyConvRule(v)$ \tcp{defined on page~\pageref{proc: ApplyConvRule}}}
  \uElseIf{$\rho \in \{\rTrans,\rTransP\}$}{
     $\ApplyTransRule(\rho,v)$\tcp*[l]{defined on page~\pageref{proc: ApplyTransRule}}
     \If{$\Status(v) = \{\Incomplete,\Unsat,\Sat\}$}{$\PropagateStatus(v)$, \Return}
  }
  \Else{
     let $X_1$, \ldots, $X_k$ be the possible conclusions of the rule\;
     \lIf{$\rho \in \{\rH,\rHP,\rV\}$}{$Y := \RFormulas(v)$\\}
     \lElse{$Y := \RFormulas(v) \cup \{\textrm{the principal formula of $\rho$}\}$}\;
     \lForEach{$1 \leq i \leq k$}{$\ConToSucc(v, \NonState, \SType(v), \Null, X_i, Y, \DFormulas(v))$}
  }

$\Status(v) := \Expanded$\;

\BlankLine
\ForEach{successor $w$ of $v$ with $\Status(w) \notin \{\Incomplete,\Unsat,\Sat\}$}{
   \lIf{$\TUnsat(w)$}{$\Status(w) := \Unsat$\\}
   \ElseIf{$\Type(w) = \NonState$}{
	$v_0 := \StatePred(w)$,\ \ 
	$v_1 := \AfterTransPred(w)$\;

	\uIf{$\SType(v_0) = \Simple$}{
	   let $\E R.C$ be the form of $\CELabel(v_1)$\;
	   $X := \Trans_\mR(\Label(w),R^-) \setminus \AFormulas(v_0)$
	}
	\Else{
	   let $a\!:\!\E R.C$ be the form of $\CELabel(v_1)$\;
	   $X := \Trans_\mR(\Label(w),R^-,a) \setminus \AFormulas(v_0)$
	}

	\If{$X \neq \emptyset$}{
	   \uIf{$\ConvMethod(v_0) = 0$}{
		$\FmlsRC(v_0) := \FmlsRC(v_0) \cup X$\;
		\lIf{$X \cap \DFormulas(v_0) \neq \emptyset$}
		   {$\Status(v_0) := \Unsat$, \Return}
	   }
	   \lElseIf{$X \cap \DFormulas(v_0) \neq \emptyset$}{$\Status(w) := \Unsat$\\}
	   \Else{
		$\AltFmlSetsSCP(v_1) := \AltFmlSetsSCP(v_1) \cup \{X\}$\;
		$\Status(w) := \Incomplete$
	   }
	}
   }
   \lElseIf{$\TSat(w)$}{$\Status(w) := \Sat$}
}

$\UpdateStatus(v)$\;
\lIf{$\Status(v) \in \{\Incomplete,\Unsat,\Sat\}$}{$\PropagateStatus(v)$}
\end{procedure}
\VSpace{-3ex}
\end{figure*}


\begin{figure*}[h!]
\begin{procedure}[H]
\caption{ApplyConvRule($v$)\label{proc: ApplyConvRule}}
\GlobalData{a rooted graph $(V,E,\nu)$.}
\Purpose{applying the rule $\rConv$ to the node $v$.}

\BlankLine

     let $w$ be the only successor of $v$,\ \ 
     $E := E \setminus \{(v,w)\}$\;

     \uIf{$\ConvMethod(w) = 0$}{
	   $newLabel := \Label(v) \cup \FmlsRC(w)$\;
	   $\ConToSucc(v,\NonState,\SType(v),\Null, newLabel, \RFormulas(v), \DFormulas(v))$
     }
     \Else{
	   let $\{\varphi_1\}$, \ldots, $\{\varphi_n\}$ be all the singleton sets belonging to $\AltFmlSetsSC(w)$, and let $remainingSetsSC$ be the set of all the remaining sets\;
	   \ForEach{$1 \leq i \leq n$}{
		$newLabel := \Label(v) \cup \{\varphi_i\}$,\ \ 
		$newDFmls := \DFormulas(v) \cup \{\varphi_j \mid 1 \leq j < i\}$\;
		$\ConToSucc(v,\NonState,\SType(v),\Null, newLabel, \RFormulas(v), newDFmls)$
	   }
	   $Y := \{\varphi_i \mid 1 \leq i \leq n\}$\;
	   \ForEach{$X \in remainingSetsSC$}{
		$\ConToSucc(v,\NonState,\SType(v),\Null, \Label(v) \cup X, \RFormulas(v), \DFormulas(v) \cup Y)$
	   }
     }

\end{procedure}


\bigskip

\begin{procedure}[H]
\caption{ApplyTransRule($\rho, u$)\label{proc: ApplyTransRule}}
\GlobalData{a rooted graph $(V,E,\nu)$.}
\Purpose{applying the transitional rule $\rho$, which is $\rTrans$ or $\rTransP$, to the state $u$.}

\BlankLine

let $X_1$, \ldots, $X_k$ be all the conclusions of the rule $\rho$ with $\Label(u)$ as the premise\;
\uIf{$\rho = \rTrans$}{
  let $\E R_1.C_1$, \ldots, $\E R_k.C_k$ be the corresponding principal formulas\;
  \ForEach{$1 \leq i \leq k$}{
    $v := \NewSucc(u,\NonState,\Simple,\E R_i.C_i,X_i,\emptyset, \emptyset)$\;
    $\FmlsRC(u) := \FmlsRC(u) \cup (\Trans_\mR(\Label(v),R_i^-) \setminus \AFormulas(u))\}$
  }
}
\Else{
  let $a_1\!:\!\E R_1.C_1$, \ldots, $a_k\!:\!\E R_k.C_k$ be the corresponding principal formulas\;
  \ForEach{$1 \leq i \leq k$}{
    $v := \NewSucc(u,\NonState,\Simple,a_i\!:\!\E R_i.C_i,X_i,\emptyset, \emptyset)$\;
    $\FmlsRC(u) := \FmlsRC(u) \cup (\Trans_\mR(\Label(v),R_i^-,a_i) \setminus \AFormulas(u)\})$
  }
}

\lIf{$\FmlsRC(u) \cap \DFormulas(u) \neq \emptyset$}{$\Status(u) := \Unsat$}\;

\BlankLine
\lWhile{$\Status(u) \neq \Unsat$ and there exists a node $w$ in the local graph of $u$ such that $\Status(w) = \Unexpanded$ and a unary rule $\rho \neq \rFormingState$ is applicable to~$w$}{$\Apply(\rho,w)$}\;

\BlankLine
\If{$\Status(u) \neq \Unsat$}{
  \lIf{$\FmlsRC(u) \neq \emptyset$}{$\Status(u) := \Incomplete$\\}
  \lElse{$\ConvMethod(u) := 1$}
}
\end{procedure}
\VSpace{-3ex}
\end{figure*}


\begin{figure*}[ht!]
\begin{function}[H]
\caption{Tableau($\mR, \mT, \mA$)\label{proc: Tableau}}
\Input{a knowledge base $(\mR,\mT,\mA)$ in NNF in the logic \SHI.}
\GlobalData{a rooted graph $(V,E,\nu)$.}
\BlankLine

$X := \mA \cup \{(a\!:\!C) \mid C \in \mT$ and $a$ is an individual occurring in $\mA\}$\;
$\nu := \NewSucc(\Null,\NonState,\Complex,\Null,X, \emptyset, \emptyset)$\;
\lIf{$\TUnsat(\nu)$}{$\Status(\nu) := \Unsat$\\}
\lElseIf{$\TSat(\nu)$}{$\Status(\nu) := \Sat$}\;

\BlankLine
\While{$(v := \ToExpand()) \neq \Null$}{
choose a tableau rule $\rho$ different from $\rConv$ and applicable to $v$\;
$\Apply(\rho, v)$\tcp*[l]{defined on page \pageref{proc: Apply}}
} 

\Return $(V,E,\nu)$
\end{function}


\begin{procedure}[H]
\caption{UpdateStatus($v$)\label{proc: UpdateStatus}}
\GlobalData{a rooted graph $(V,E,\nu)$.}
\Input{a node $v \in V$ with $\Status(v) = \Expanded$.}
\BlankLine

\uIf{$\Kind(v) = \OrNode$}{
   \lIf{some successors of $v$ have status $\Sat$}{$\Status(v) := \Sat$\\}
   \lElseIf{all successors of $v$ have status $\Unsat$}{$\Status(v) := \Unsat$\\}
   \ElseIf{every successor of $v$ has status $\Incomplete$ or $\Unsat$}{
	\uIf{$v$ has a successor $w$ such that $\Type(w) = \State$}{
	   \tcp{$w$ is the only successor of~$v$}
	   $\Apply(\rConv,v)$
	}
	\lElse{$\Status(v) := \Incomplete$}
   }
}
\Else(\tcp*[h]{$\Kind(v) = \AndNode$}){
   \lIf{all successors of $v$ have status $\Sat$}{$\Status(v) := \Sat$\\}
   \lElseIf{some successors of $v$ have status $\Unsat$}{$\Status(v) := \Unsat$\\}
   \ElseIf{$v$ has a successor $w$ with $\Status(w) = \Incomplete$}{
	$\AltFmlSetsSC(v) := \AltFmlSetsSCP(w)$,\ \ 
	$\Status(v) := \Incomplete$ 
   }
}
\end{procedure}



\begin{procedure}[H]
\caption{PropagateStatus($v$)\label{proc: PropagateStatus}}
\GlobalData{a rooted graph $(V,E,\nu)$.}
\Input{a node $v \in V$ with $\Status(v) \in \{\Incomplete,\Unsat,\Sat\}$.}
\BlankLine

\ForEach{predecessor $u$ of $v$ with $\Status(u) = \Expanded$}{
   $\UpdateStatus(u)$\;
   \lIf{$\Status(u) \in \{\Incomplete,\Unsat,\Sat\}$}{$\PropagateStatus(u)$}
}
\end{procedure}
\VSpace{-3ex}
\end{figure*}
} 


Our calculus \CSHI\ for the description logic \SHI will be specified, amongst others, by a finite set of tableau rules, which are used to expand nodes of tableaux. A~{\em tableau rule} is specified with the following information:
\begin{itemize}
\item the kind of the rule: an ``and''-rule or an ``or''-rule
\item the conditions for applicability of the rule (if any)
\item the priority of the rule 
\item the number of successors of a node resulting from applying the rule to it, and the way to compute their contents.
\end{itemize}

Tableau rules are usually written downwards, with a set of formulas above the line as the {\em premise}, which represents the label of the node to which the rule is applied, and a number of sets of formulas below the line as the {\em (possible) conclusions}, which represent the labels of the successor nodes resulting from the application of the rule.
Possible conclusions of an ``or''-rule are separated by $\mid$, while conclusions of an ``and''-rule are separated by~$\&$. If a rule is a unary rule (i.e.\ a rule with only one possible conclusion) or an ``and''-rule then its conclusions are ``firm'' and we ignore the word ``possible''.
The meaning of an ``or''-rule is that if the premise is satisfiable w.r.t.\ $\mR$ and $\mT$ then some of the possible conclusions are also satisfiable w.r.t.\ $\mR$ and $\mT$, while the meaning of an ``and''-rule is that if the premise is satisfiable w.r.t.\ $\mR$ and $\mT$ then all of the conclusions are also satisfiable w.r.t.\ $\mR$ and~$\mT$.

Such a representation gives only a part of the specification of the rules.

We write $X,\varphi$ or $\varphi,X$ to denote $X \cup \{\varphi\}$, and write $X, Y$ to denote $X \cup Y$. 
Our {\em tableau calculus \CSHI} for \SHI w.r.t.\ the RBox $\mR$ and the TBox $\mT$ consists of rules which are partially specified in Table~\ref{table: CSHI} together with two special rules $\rFormingState$ and $\rConv$. 

The rules $\rTrans$ and $\rTransP$ are the only ``and''-rules and the only {\em transitional rules}. 
The other rules of \CSHI are ``or''-rules, which are also called {\em static rules}. 
The transitional rules are used to expand states of tableaux, while the static rules are used to expand non-states of tableaux. 

For any rule of \CSHI except $\rFormingState$ and $\rConv$, the distinguished formulas of the premise are called the {\em principal formulas} of the rule. The rules $\rFormingState$ and $\rConv$ have no principal formulas.
As usually, we assume that, for each rule of \CSHI described in Table~\ref{table: CSHI}, the principal formulas are not members of the set $X$ which appears in the premise of the rule.

Expanding a non-state $v$ of a tableau by a static rule $\rho \in \{\rAnd,\rOr,\rAndP,\rOrP\}$ which uses $\varphi$ as the principal formula, we put $\varphi$ into the set $\RFormulas(w)$ of each successor $w$ of $v$. 
Recall that $\RFormulas(w)$ is called the set of the reduced formulas of~$w$. If $w$ is a non-state, $v_1 = \AfterTransPred(w)$ and $v_1, v_2, \ldots, v_k = w$ is the path (of non-states) from $v_1$ to $w$, then an occurrence $\psi \in \RFormulas(w)$ means there exists $1 \leq i < k$ such that $\psi \in \Label(v_i)$ and $\psi$ has been reduced at $v_i$. After that reduction, $\psi$ was put into $\RFormulas(v_{i+1})$ and propagated to $\RFormulas(v_k)$. 

Expanding a simple (resp.\ complex) state $v$ of a tableau by the transitional rule $\rTrans$ (resp.\ $\rTransP$), each successor $w_i$ of $v$ is created due to a corresponding principal formula $\E R_i.C_i$ (resp.\ $a_i\!:\!\E R_i.C_i$) of the rule, and $\RFormulas(w)$ is set to the empty set.

For any state $w$, every predecessor $v$ of $w$ is always a non-state. Such a node $v$ was expanded and connected to $w$ by the static rule $\rFormingState$. The nodes $v$ and $w$ correspond to the same element of the domain of the interpretation under construction. In other words, the rule $\rFormingState$ ``transforms'' a non-state to a state. It guarantees that, if $\BeforeFormingState(v)$ holds then $v$ has exactly one successor, which is a state.

The rule $\rConv$ used for dealing with converses will be discussed shortly.

The priorities of the rules of \CSHI are as follows (the bigger, the stronger): 
$\rAnd$, $\rAndP$, $\rH$, $\rHP$, $\rV$:~5; 
$\rOr$, $\rOrP$:~4;
$\rFormingState$:~3;
$\rTrans$, $\rTransP$:~2; 
$\rConv$:~1.

The conditions for applying a rule $\rho \neq \rConv$ to a node $v$ are as follows: 
\begin{itemize}
\item the rule has $\Label(v)$ as the premise (thus, the rules $\rAnd$, $\rOr$, $\rH$, $\rTrans$ are applicable only to simple nodes, and the rules $\rAndP$, $\rOrP$, $\rHP$, $\rV$, $\rTransP$ are applicable only to complex nodes)
\item all the conditions accompanying with $\rho$ in Table~\ref{table: CSHI} are satisfied
\item if $\rho$ is a transitional rule then $\Type(v) = \State$
\item if $\rho$ is a static rule then $\Type(v) = \NonState$ and 
  \begin{itemize}
  \item if $\rho \in \{\rAnd, \rOr, \rAndP, \rOrP\}$ then the principal formula of $\rho$ does not belong to $\RFormulas(v)$, else if $\rho \in \{\rH,\rHP,\rV\}$ then the formula occurring in the conlusion but not in the premise of $\rho$ does not belong to $\AFormulas(v)$
  \item no static rule with a higher priority is applicable to~$v$.
  \end{itemize}
\end{itemize}


We now explain the ways of dealing with converses, i.e., with inverse roles. 

Consider the case when $\Type(v) = \State$, $\SType(v) = \Simple$, $\E R.C \in \Label(v)$ and $v$ corresponds to an element $x_v \in \Delta^\mI$ of the interpretation $\mI$ under construction. We need to realize the formulas of $\Label(v)$ at $v$ so that $x_v \in (\Label(v))^\mI$. The formula $\E R.C$ is realized at $v$ by making a transition from $v$ to $w$ with $\Label(w) = \{C\} \cup \Trans_\mR(\Label(v),R) \cup \mT$. The node $w$ corresponds to an element $x_w \in \Delta^\mI$ such that $(x_v,x_w) \in R^\mI$ and $x_w \in C^\mI$. If at some later stage we need to make $x_w \in (\V R^-.D)^\mI$ (for example, because $(\V R^-.D) \in \Label(w)$) then we need to make $x_v \in D^\mI$, and hence we need to add $D$ to $\Label(v)$ as a requirement to be realized at $v$ if $D \notin \AFormulas(v)$. 
Similarly, if at some later stage we need to make $x_w \in (\V S.D)^\mI$, where $R^- \sqsubseteq_\mR S$ and $S \circ S \sqsubseteq S \in \Ext(\mR)$, then we need to make $x_v \in (\V S.D)^\mI$, and hence we need to add $\V S.D$ to $\Label(v)$ as a requirement to be realized at $v$ if $\V S.D \notin \AFormulas(v)$. 
  \begin{itemize}
  \item If $x_v \in D^\mI$ (where $D$ may be of the form $\V S.D'$) is a requirement but $D \notin \AFormulas(v)$ then we record this by setting $\ConvMethod(v) := 0$ and add $D$ to the set $\FmlsRC(v)$. If $\FmlsRC(v) \cap \DFormulas(v) \neq \emptyset$ then the requirements at $v$ are unrealizable and we set $\Status(v) := \Unsat$ (which means $\FullLabel(v)$ is unsatisfiable w.r.t.\ $\mR$ and~$\mT$). 
If $\FmlsRC(v) \neq \emptyset$ and $\FmlsRC(v) \cap \DFormulas(v) = \emptyset$ then we set $\Status(v) := \Incomplete$, which means the set $\Label(v)$ should be extended with $\FmlsRC(v)$ if $v$ will be used. 

  \item Consider the case when the computed set $\FmlsRC(v)$ is empty. In this case, we set $\ConvMethod(v) := 1$. Each node $w_i$ in the local graph of $w$ is an ``or''-descendant of $w$ and corresponds to the same $x_w \in \Delta^\mI$ (for example, if $C_1 \mor C_2 \in \Label(w)$ then we may make $w$ an ``or''-node with two successors $w_1$ and $w_2$ such that $C_1 \in \Label(w_1)$ and $C_2 \in \Label(w_2)$). 
     \begin{itemize}
     \item Consider the case $(\V R^-.D) \in \Label(w_i)$. Thus, $x_w \in (\V R^-.D)^\mI$ is one of possibly many alternative requirements (because $w_i$ is one of possibly many ``or''-descendants of $w$). If $w_i$ should be selected for representing $w$ and $D \notin \AFormulas(v)$ then we should add $D$ to $\Label(v)$. If $D \in \DFormulas(v)$ then we set $\Status(w_i) := \Unsat$, which means the ``combination'' of $v$ and $w_i$ is unsatisfiable w.r.t.~$\mR$ and~$\mT$. 
     \item Consider the case when $(\V S.D) \in \Label(w_i)$, $R^- \sqsubseteq_\mR S$ and $S \circ S \sqsubseteq S \in \Ext(\mR)$. Thus, $x_w \in (\V S.D)^\mI$ is one of possibly many alternative requirements (because $w_i$ is one of possibly many ``or''-descendants of $w$). If $w_i$ should be selected for representing $w$ and $\V S.D \notin \AFormulas(v)$ then we should add $\V S.D$ to $\Label(v)$. If $\V S.D \in \DFormulas(v)$ then we set $\Status(w_i) := \Unsat$, which means the ``combination'' of $v$ and $w_i$ is unsatisfiable w.r.t.~$\mR$ and~$\mT$. 
     \end{itemize}
If, for $X = \Trans_\mR(\Label(w_i),R^-) \setminus \AFormulas(v)$, we have that $X \neq \emptyset$ and $X \cap \DFormulas(v) = \emptyset$, then we add $X$ (as an element) to the set $\AltFmlSetsSCP(w)$ and set $\Status(w_i) := \Incomplete$, which means that, if the ``or''-descendant $w_i$ should be selected for representing $w$ then $X$ should be added (as a set) to $\Label(v)$.
  \end{itemize}


\ShortVersion{
The case when $\Type(v) = \State$, $\SType(v) = \Complex$ and \mbox{$a\!:\!\E R.C \in \Label(v)$} can be dealt with in a similar way. See~\cite{nSHI-long} for details.
}
\LongVersion{
Now consider the case when $\Type(v) = \State$, $\SType(v) = \Complex$ and \mbox{$a\!:\!\E R.C \in \Label(v)$}. It is very similar to the previous one. 
We need to satisfy (the ABox) $\Label(v)$ in the interpretation $\mI$ under construction. To satisfy the formula $a\!:\!\E R.C$ in $\mI$ we make a transition from $v$ to $w$ with $\Label(w) = \{C\} \cup \Trans_\mR(\Label(v),a,R) \cup \mT$. The node $w$ corresponds to an element $x_w \in \Delta^\mI$ such that $(a^\mI,x_w) \in R^\mI$ and $x_w \in C^\mI$. If at some later stage we need to make $x_w \in (\V R^-.D)^\mI$ (for example, because $(\V R^-.D) \in \Label(w)$) then we need to make $a^\mI \in D^\mI$, and hence we need to add $a\!:\!D$ to $\Label(v)$ as a requirement to be realized at $v$ if $(a\!:\!D) \notin \AFormulas(v)$. 
Similarly, if at some later stage we need to make $x_w \in (\V S.D)^\mI$, where $R^- \sqsubseteq_\mR S$ and $S \circ S \sqsubseteq S \in \Ext(\mR)$, then we need to make $a^\mI \in (\V S.D)^\mI$, and hence we need to add $a\!:\!\V S.D$ to $\Label(v)$ as a requirement to be realized at $v$ if $a\!:\!\V S.D \notin \AFormulas(v)$. 
  \begin{itemize}
  \item If $a^\mI \in D^\mI$ (where $D$ may be of the form $\V S.D'$) is a requirement but $(a\!:\!D) \notin \AFormulas(v)$ then we record this by setting $\ConvMethod(v) := 0$ and add $a\!:\!D$ to the set $\FmlsRC(v)$. If $\FmlsRC(v) \cap \DFormulas(v) \neq \emptyset$ then the requirements at $v$ are unrealizable and we set $\Status(v) := \Unsat$ (which means $\FullLabel(v)$ is unsatisfiable w.r.t.\ $\mR$ and~$\mT$). 
If $\FmlsRC(v) \neq \emptyset$ and $\FmlsRC(v) \cap \DFormulas(v) = \emptyset$ then we set $\Status(v) := \Incomplete$, which means the set $\Label(v)$ should be extended with $\FmlsRC(v)$ if $v$ will be used. 

  \item Consider the case when the computed set $\FmlsRC(v)$ is empty. In this case, we set $\ConvMethod(v) := 1$. Each node $w_i$ in the local graph of $w$ is an ``or''-descendant of $w$ and corresponds to the same $x_w \in \Delta^\mI$. 
     \begin{itemize}
     \item Consider the case $(\V R^-.D) \in \Label(w_i)$. Thus, $x_w \in (\V R^-.D)^\mI$ is one of possibly many alternative requirements (because $w_i$ is one of possibly many ``or''-descendants of $w$). If $w_i$ should be selected for representing $w$ and $(a\!:\!D) \notin \AFormulas(v)$ then we should add $a\!:\!D$ to $\Label(v)$. If $(a\!:\!D) \in \DFormulas(v)$ then we set $\Status(w_i) := \Unsat$, which means the ``combination'' of $v$ and $w_i$ is unsatisfiable w.r.t.~$\mR$ and~$\mT$. 
     \item Consider the case when $(\V S.D) \in \Label(w_i)$, $R^- \sqsubseteq_\mR S$ and $S \circ S \sqsubseteq S \in \Ext(\mR)$. Thus, $x_w \in (\V S.D)^\mI$ is one of possibly many alternative requirements (because $w_i$ is one of possibly many ``or''-descendants of $w$). If $w_i$ should be selected for representing $w$ and $(a\!:\!\V S.D) \notin \AFormulas(v)$ then we should add $a\!:\!\V S.D$ to $\Label(v)$. If $(a\!:\!\V S.D) \in \DFormulas(v)$ then we set $\Status(w_i) := \Unsat$, which means the ``combination'' of $v$ and $w_i$ is unsatisfiable w.r.t.~$\mR$ and~$\mT$. 
     \end{itemize}

If, for $X = \Trans_\mR(\Label(w_i),R^-,a) \setminus \AFormulas(v)\}$, we have that $X \neq \emptyset$ and $X \cap \DFormulas(v) = \emptyset$, then we add $X$ (as an element) to the set $\AltFmlSetsSCP(w)$ and set $\Status(w_i) := \Incomplete$, which means that, if the ``or''-descendant $w_i$ should be selected for representing $w$ then $X$ should be added (as a set) to $\Label(v)$.
  \end{itemize}
} 


\ShortVersion{When a node $w$ gets status $\Incomplete$, $\Unsat$ or $\Sat$, the status of every predecessor $v$ of $w$ will be updated as shown in procedure $\UpdateStatus(v)$ defined in~\cite[page~10]{nSHI-long}.} 
\LongVersion{When a node $w$ gets status $\Incomplete$, $\Unsat$ or $\Sat$, the status of every predecessor $v$ of $w$ will be updated as shown in procedure $\UpdateStatus(v)$ defined on page~\pageref{proc: UpdateStatus}.} 
In particular:
\begin{itemize}
\item If $\Type(w) = \State$ and $\Status(w) = \Incomplete$ then $\BeforeFormingState(v)$ holds and $w$ is the only successor of $v$. In this case, the edge $(v,w)$ will be deleted and the node $v$ will be re-expanded by the converse rule $\rConv$ as shown in procedure $\ApplyConvRule$ given\LongVersion{ on page~\pageref{proc: ApplyConvRule}}\ShortVersion{ in~\cite[page~9]{nSHI-long}}. For the subcase $\ConvMethod(w) = 0$, we connect $v$ to a node with label $\Label(v) \cup \FmlsRC(w)$. Consider the subcase when $\ConvMethod(w) = 1$. Let $\AltFmlSetsSC(w) = \{\{\varphi_1\}$, \ldots, $\{\varphi_n\}$, $Z_1$, \ldots, $Z_m\}$, where $Z_1$, \ldots, $Z_m$ are non-singleton sets. We connect $v$ to successors $w_1$, \ldots, $w_{n+m}$ such that: for $1 \leq i \leq n$, $\Label(w_i) = \Label(v) \cup \{\varphi_i\}$, and for $n+1 \leq i \leq n+m$, $\Label(w_i) = \Label(v) \cup Z_i$. To restrict the search space, for $1 \leq i \leq n$, we add all $\varphi_j$ with $1 \leq j < i$ to $\DFormulas(w_i)$. This can be read as: at $v$ either allow to have $\varphi_1$ (by adding it to the attribute $\Label$), or disallow $\varphi_1$ (by adding it to the attribute $\DFormulas$) and allow $\varphi_2$, or disallow $\varphi_1$, $\varphi_2$ and allow $\varphi_3$, and so on. Similarly, for $n+1 \leq i \leq n+m$, we add all $\varphi_j$ with $1 \leq j \leq n$ to $\DFormulas(w_i)$.

\item If $\Type(v) = \State$ (i.e.\ $\Kind(v) = \AndNode$) and $v$ has a successor $w$ such that $\Status(w) = \Incomplete$ then we set $\AltFmlSetsSC(v) := \AltFmlSetsSCP(w)$ and set $\Status(v) := \Incomplete$.
\end{itemize}


Application of a tableau rule $\rho$ to a node $v$ is specified by procedure $\Apply(\rho, v)$ given\LongVersion{ on page \pageref{proc: Apply}}\ShortVersion{ in~\cite[page~8]{nSHI-long}}. This procedure uses procedures $\ApplyConvRule$ and $\ApplyTransRule$ given\LongVersion{ on page~\pageref{proc: ApplyConvRule}}\ShortVersion{ in~\cite[page~9]{nSHI-long}}. \LongVersion{Auxiliary functions are defined on page~\pageref{proc: NewSucc}. }Procedures used for updating and propagating statuses of nodes are defined\LongVersion{ on page~\pageref{proc: Tableau}}\ShortVersion{ in~\cite[page~10]{nSHI-long}}. 
The main function $\Tableau(\mR,\mT,\mA)$ is also defined\LongVersion{ on page~\pageref{proc: Tableau}}\ShortVersion{ in~\cite[page~10]{nSHI-long}}. It returns a rooted ``and-or'' graph called a {\em \CSHI-tableau} for the knowledge base $(\mR,\mT,\mA)$. 
The root of the graph is a complex node $\nu$ with $\Label(\nu) = \mA \cup \{(a\!:\!C)$ $\mid C \in \mT$ and $a$ is an individual occurring in $\mA\}$. \LongVersion{Also notice that trivial unsatisfiability and satisfiability are checked immediately for each newly created node.} 


\LongVersion{
\newcommand{\myhline}{\\[0.4ex] \hline \\[-1.6ex]}
\begin{figure}
\begin{center}
\begin{tabular}{c}
\begin{scriptsize}
\xymatrix{
*+[F]{\begin{tabular}{c}
	(1): $\rOrP$, or
	\myhline
	$a\!:\!F$, $L(a,b)$,\\
	$b\!:\!\E L.\lnot I$,\\
	$a\!:\!\varphi$, $b\!:\!\varphi$
      \end{tabular}}
\ar@{->}[r]
\ar@{->}[d]
&
*+[F]{\begin{tabular}{c}
	(3): $\rAndP$
	\myhline
	$a\!:\!F$, $L(a,b)$,\\
	$b\!:\!\E L.\lnot I$, 
	$b\!:\!\varphi$,\\
	$a\!:\!I \mand \V P.F$
      \end{tabular}}
\ar@{->}[r]
&
*+[F]{\begin{tabular}{c}
	(4): $\rHP$
	\myhline
	$a\!:\!F$, $L(a,b)$,\\
	$b\!:\!\E L.\lnot I$, 
	$b\!:\!\varphi$,\\
	$a\!:\!I$, $a\!:\!\V P.F$
      \end{tabular}}
\ar@{->}[r]
&
*+[F]{\begin{tabular}{c}
	(5): $\rV$
	\myhline
	$a\!:\!F$, $L(a,b)$,\\
	$b\!:\!\E L.\lnot I$, 
	$b\!:\!\varphi$,\\
	$a\!:\!I$, $a\!:\!\V P.F$,\\
	$a\!:\!\V L.F$
      \end{tabular}}
\ar@{->}[d]
\\
*+[F]{\begin{tabular}{c}
	(2)
	\myhline
	$a\!:\!F$, $a\!:\!\lnot F$,\\
	\ldots\\
	$\Unsat$
      \end{tabular}}
&
*+[F]{\begin{tabular}{c}
	(9): $\rAndP$
	\myhline
	$a\!:\!F$, $L(a,b)$,\\
	$b\!:\!\E L.\lnot I$,\\ 
	$a\!:\!I$, $a\!:\!\V P.F$,\\
	$a\!:\!\V L.F$,\\ 
	$b\!:\!F$, $b\!:\!\V P.F$,\\ 
	$b\!:\!\V L.F$,\\
        $b\!:\!I \mand \V P.F$
      \end{tabular}}
\ar@{->}[ld]
&
*+[F]{\begin{tabular}{c}
	(7): $\rOrP$, or
	\myhline
	$a\!:\!F$, $L(a,b)$,\\
	$b\!:\!\E L.\lnot I$, 
	$b\!:\!\varphi$,\\
	$a\!:\!I$, $a\!:\!\V P.F$,\\
	$a\!:\!\V L.F$,\\
	$b\!:\!F$, $b\!:\!\V P.F$,\\  
	$b\!:\!\V L.F$
      \end{tabular}}
\ar@{->}[l]
\ar@{->}[dr]
&
*+[F]{\begin{tabular}{c}
	(6): $\rHP$
	\myhline
	$a\!:\!F$, $L(a,b)$,\\
	$b\!:\!\E L.\lnot I$, 
	$b\!:\!\varphi$,\\
	$a\!:\!I$, $a\!:\!\V P.F$,\\
	$a\!:\!\V L.F$,\\ 
	$b\!:\!F$, $b\!:\!\V P.F$
      \end{tabular}}
\ar@{->}[l]
\\
*+[F]{\begin{tabular}{c}
	(10):\\
	\rFormingState
	\myhline
	$a\!:\!F$, $L(a,b)$,\\
	$b\!:\!\E L.\lnot I$,\\ 
	$a\!:\!I$, $a\!:\!\V P.F$,\\
	$a\!:\!\V L.F$,\\ 
	$b\!:\!F$, $b\!:\!\V P.F$,\\ 
	$b\!:\!\V L.F$,\\
	$b\!:\!I$, $b\!:\!\V P.F$
      \end{tabular}}
\ar@{->}[r]
&
*+[F=]{\begin{tabular}{c}
	(11): $\rTransP$
	\myhline
	$a\!:\!F$, $L(a,b)$,\\
	$b\!:\!\E L.\lnot I$,\\ 
	$a\!:\!I$, $a\!:\!\V P.F$,\\
	$a\!:\!\V L.F$,\\ 
	$b\!:\!F$, $b\!:\!\V P.F$,\\ 
	$b\!:\!\V L.F$,\\
	$b\!:\!I$, $b\!:\!\V P.F$
      \end{tabular}}
\ar@{->}[r]
&
*+[F]{\begin{tabular}{c}
	(12): $\rOr$, or
	\myhline
	$\lnot I$, $F$, $\V P.F$,\\
	$\varphi$
      \end{tabular}}
\ar@{->}[dl]
\ar@{->}[dr]
&
*+[F]{\begin{tabular}{c}
	(8): $\rBotP$
	\myhline
	$b\!:\!F$, $b\!:\!\lnot F$,\\
	\ldots\\
	$\Unsat$
      \end{tabular}}
\\
*+[F]{\begin{tabular}{c}
	(15): $\rBot$
	\myhline
	$\lnot I$, $F$, $\V P.F$,\\
	$I$, $\V P.F$\\
	$\Unsat$
      \end{tabular}}
&
*+[F]{\begin{tabular}{c}
	(14): $\rAnd$
	\myhline
	$\lnot I$, $F$, $\V P.F$,\\
	$I \mand \V P.F$
      \end{tabular}}
\ar@{->}[l]
&
&
*+[F]{\begin{tabular}{c}
	(13): $\rBot$
	\myhline
	$\lnot I$, $F$, $\V P.F$,\\
	$\lnot F$\\
	$\Unsat$
      \end{tabular}}
} 
\end{scriptsize}
\end{tabular}
\end{center}
\caption{An ``and-or'' graph for the knowledge base $(\mR,\mT,\mA')$, where $\mR = \{L \sqsubseteq P$, $P \circ P \sqsubseteq P\}$, $\mT = \{\varphi\}$, $\mA' = \{a\!:\!F$, $L(a,b)$, $b\!:\!\E L.\lnot I\}$, and $\varphi = \lnot F \mor (I \mand \V P.F)$. 
In each node, we display the name of the rule expanding the node and the formulas of the node's label. The node (11) is the only state. We have, for example, $\StatePred((15)) = (11)$, $\AfterTransPred((15)) = (12)$ and $\CELabel((12)) = b\!:\!\E L.\lnot I$.}
\label{fig: example}
\end{figure}

\newcommand{\perfect}{\mathit{perfect}}
\newcommand{\interesting}{\mathit{interesting}}
\newcommand{\link}{\mathit{link}}
\newcommand{\Path}{\mathit{path}}

\begin{example} \label{example1}
This is an example about web pages, taken from~\cite{NguyenS10TCCI} and adapted to our calculus. Let
\begin{eqnarray*}
\mR & = & \{\link \sqsubseteq \Path,\; \Path \circ \Path \sqsubseteq \Path\}\\
\mT & = & \{\perfect \sqsubseteq \interesting \mand \V\Path.\perfect\}\\
\mA & = & \{ a\!:\!\perfect, \link(a,b) \}
\end{eqnarray*}
It can be shown that $b$ is an instance of the concept $\V\link.\interesting$ w.r.t.\ the knowledge base $(\mR,\mT,\mA)$, i.e., for every
model $\mI$ of $(\mR,\mT,\mA)$, we have that $b^\mI \in (\V\link.\interesting)^\mI$. To prove this one can show that the knowledge base
$(\mR,\mT,\mA')$, where $\mA' = \mA \cup \{b\!:\!\E\link.\lnot\interesting\}$, is unsatisfiable.
As abbreviations, let $L = \link$, $P = \Path$, $I = \interesting$, $F = \perfect$, and $\varphi = \lnot F \mor (I \mand \V P.F)$. We have
\begin{eqnarray*}
\mR & = & \{L \sqsubseteq P,\; P \circ P \sqsubseteq P\} \\
\mT & = & \{\varphi\} \textrm{ (in NNF).} \\
\mA' & = & \{a\!:\!F,\; L(a,b),\; b\!:\!\E L.\lnot I\}. 
\end{eqnarray*}
An ``and-or'' graph for $(\mR,\mT,\mA')$ is presented in Figure~\ref{fig: example}. 
\myEnd
\end{example}
} 


\LongVersion{
\begin{example}\label{ex:andorgraph}
Let 
\begin{eqnarray*}
\mR & = & \{r \sqsubseteq s,\ \ r^- \sqsubseteq s,\ \ s \circ s \sqsubseteq s\}\\
\mT & = & \{\E r.(A \mand \V s.\lnot A)\}\\
\mA & = & \{ a\!:\!\top \}.
\end{eqnarray*}
In Figures~\ref{fig: example2} and \ref{fig: example2-II} we give an ``and-or'' graph for the knowledge base $(\mR,\mT,\mA)$. 
The nodes are numbered when created and are expanded using~DFS (depth-first search). At the end the root receives status $\Unsat$. Therefore, by Theorem~\ref{theorem: s-c}, $(\mR,\mT,\mA)$ is unsatisfiable. As a consequence, $(\mR,\mT,\emptyset)$ is also unsatisfiable.  
\myEnd
\end{example}

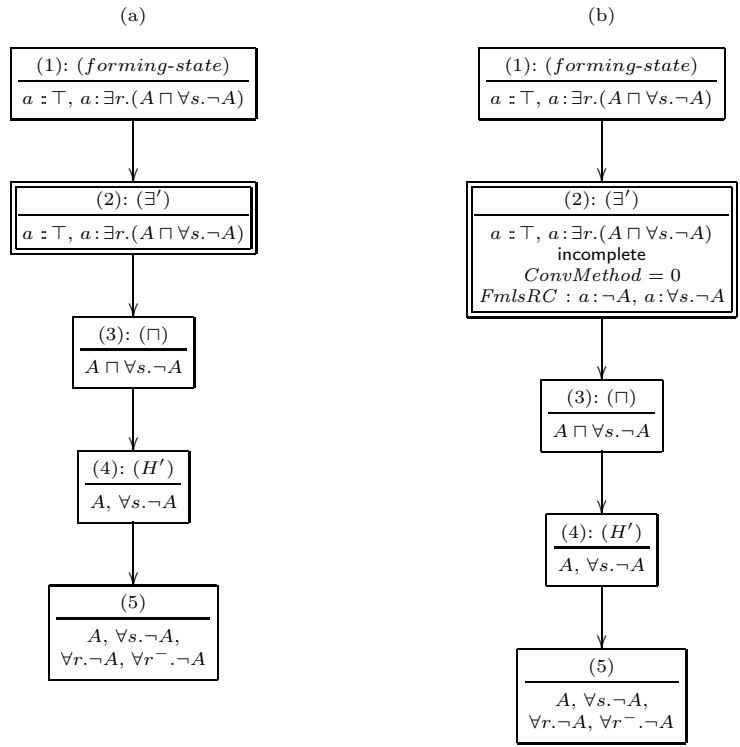
\begin{figure}
\begin{center}
\begin{tabular}{c}
\begin{scriptsize}
\begin{tabular}{c@{\extracolsep{10em}}c}
(a) & (b) \\
\\
\xymatrix{
*+[F]{\begin{tabular}{c}
	(1): $\rFormingState$
	\myhline
	$a:\!:\!\top$, $a\!:\!\E r.(A \mand \V s.\lnot A)$
      \end{tabular}}
\ar@{->}[d] 
\\
*+[F=]{\begin{tabular}{c}
	(2): $\rTransP$
	\myhline
	$a:\!:\!\top$, $a\!:\!\E r.(A \mand \V s.\lnot A)$
      \end{tabular}}
\ar@{->}[d] 
\\
*+[F]{\begin{tabular}{c}
	(3): $\rAnd$
	\myhline
	$A \mand \V s.\lnot A$
      \end{tabular}}
\ar@{->}[d]
\\
*+[F]{\begin{tabular}{c}
	(4): $\rHP$
	\myhline
	$A$, $\V s.\lnot A$
      \end{tabular}}
\ar@{->}[d]
\\
*+[F]{\begin{tabular}{c}
	(5)
	\myhline
	$A$, $\V s.\lnot A$,\\
	$\V r.\lnot A$, $\V r^-.\lnot A$
      \end{tabular}}
}
&
\xymatrix{
*+[F]{\begin{tabular}{c}
	(1): $\rFormingState$
	\myhline
	$a:\!:\!\top$, $a\!:\!\E r.(A \mand \V s.\lnot A)$
      \end{tabular}}
\ar@{->}[d] 
\\
*+[F=]{\begin{tabular}{c}
	(2): $\rTransP$
	\myhline
	$a:\!:\!\top$, $a\!:\!\E r.(A \mand \V s.\lnot A)$\\
	$\Incomplete$\\
	$\ConvMethod = 0$\\
	$\FmlsRC$ : $a\!:\!\lnot A$, $a\!:\!\V s.\lnot A$
      \end{tabular}}
\ar@{->}[d] 
\\
*+[F]{\begin{tabular}{c}
	(3): $\rAnd$
	\myhline
	$A \mand \V s.\lnot A$
      \end{tabular}}
\ar@{->}[d]
\\
*+[F]{\begin{tabular}{c}
	(4): $\rHP$
	\myhline
	$A$, $\V s.\lnot A$
      \end{tabular}}
\ar@{->}[d]
\\
*+[F]{\begin{tabular}{c}
	(5)
	\myhline
	$A$, $\V s.\lnot A$,\\
	$\V r.\lnot A$, $\V r^-.\lnot A$
      \end{tabular}}
}
\end{tabular}
\end{scriptsize}
\end{tabular}
\end{center}
\caption{\label{fig: example2}An illustration for Example~\ref{ex:andorgraph}: part I. The graph (a) is the ``and-or'' graph constructed until checking ``compatibility'' of the node (5) w.r.t. to the node (2). 
In each node, we display the name of the rule expanding the node and the formulas of the node's label. The node (2) is the only state. We have, for example, $\StatePred((5)) = (2)$, $\AfterTransPred((5)) = (3)$ and $\CELabel((3)) = a\!:\!\E r.(A \mand \V s.\lnot A)$. Checking ``compatibility'' of the node (5) w.r.t. to the node (2), $\Status((2))$ is set to $\Incomplete$ and $\FmlsRC((2))$ is set to $\{a\!:\!\lnot A, a\!:\!\V s.\lnot A\}$. This results in the graph~(b). The construction is then continued by applying the rule $\rConv$ to (1). See Figure~\ref{fig: example2-II} for the continuation.} 
\end{figure}

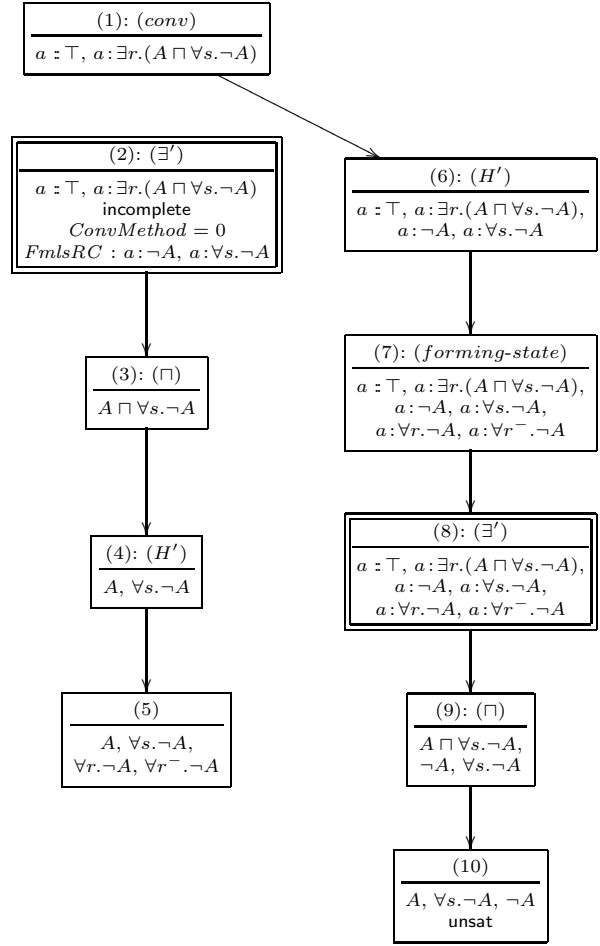
\begin{figure}
\begin{center}
\begin{tabular}{c}
\begin{scriptsize}
\xymatrix{
*+[F]{\begin{tabular}{c}
	(1): $\rConv$
	\myhline
	$a:\!:\!\top$, $a\!:\!\E r.(A \mand \V s.\lnot A)$
      \end{tabular}}
\ar@{->}[dr] 
\\
*+[F=]{\begin{tabular}{c}
	(2): $\rTransP$
	\myhline
	$a:\!:\!\top$, $a\!:\!\E r.(A \mand \V s.\lnot A)$\\
	$\Incomplete$\\
	$\ConvMethod = 0$\\
	$\FmlsRC$ : $a\!:\!\lnot A$, $a\!:\!\V s.\lnot A$
      \end{tabular}}
\ar@{->}[d] 
&
*+[F]{\begin{tabular}{c}
	(6): $\rHP$
	\myhline
	$a:\!:\!\top$, $a\!:\!\E r.(A \mand \V s.\lnot A)$,\\ 
	$a\!:\!\lnot A$, $a\!:\!\V s.\lnot A$
      \end{tabular}}
\ar@{->}[d] 
\\
*+[F]{\begin{tabular}{c}
	(3): $\rAnd$
	\myhline
	$A \mand \V s.\lnot A$
      \end{tabular}}
\ar@{->}[d]
&
*+[F]{\begin{tabular}{c}
	(7): $\rFormingState$
	\myhline
	$a:\!:\!\top$, $a\!:\!\E r.(A \mand \V s.\lnot A)$,\\ 
	$a\!:\!\lnot A$, $a\!:\!\V s.\lnot A$,\\
	$a\!:\!\V r.\lnot A$, $a\!:\!\V r^-.\lnot A$
      \end{tabular}}
\ar@{->}[d]
\\
*+[F]{\begin{tabular}{c}
	(4): $\rHP$
	\myhline
	$A$, $\V s.\lnot A$
      \end{tabular}}
\ar@{->}[d]
&
*+[F=]{\begin{tabular}{c}
	(8): $\rTransP$
	\myhline
	$a:\!:\!\top$, $a\!:\!\E r.(A \mand \V s.\lnot A)$,\\ 
	$a\!:\!\lnot A$, $a\!:\!\V s.\lnot A$,\\
	$a\!:\!\V r.\lnot A$, $a\!:\!\V r^-.\lnot A$
      \end{tabular}}
\ar@{->}[d]
\\
*+[F]{\begin{tabular}{c}
	(5)
	\myhline
	$A$, $\V s.\lnot A$,\\
	$\V r.\lnot A$, $\V r^-.\lnot A$
      \end{tabular}}
&
*+[F]{\begin{tabular}{c}
	(9): $\rAnd$
	\myhline
	$A \mand \V s.\lnot A$,\\ 
	$\lnot A$, $\V s.\lnot A$
      \end{tabular}}
\ar@{->}[d]
\\
&
*+[F]{\begin{tabular}{c}
	(10)
	\myhline
	$A$, $\V s.\lnot A$, $\lnot A$\\
	$\Unsat$
      \end{tabular}}
}
\end{scriptsize}
\end{tabular}
\end{center}
\caption{\label{fig: example2-II}An illustration for Example~\ref{ex:andorgraph}: part II. This is a fully expanded ``and-or'' graph for $(\mR,\mT,\mA)$. The node (1) is re-expanded by the rule $\rConv$. As in the part~I, in each node we display the name of the rule expanding the node and the formulas of the node's label. 
The nodes (2) and (8) are the only states. After the node (10) receives status $\Unsat$, the nodes (9)-(6) and (1) receive status $\Unsat$ in subsequent steps.
} 
\end{figure}
} 

\ShortVersion{See the long version \cite{nSHI-long} of this paper for examples of ``and-or'' graphs and a proof of the following theorem.}


\LongVersion{
Let $\closure(\mR,\mT,\mA)$ be the union of 
\begin{itemize}
\item the set of all formulas $C$ and $a\!:\!C$ such that $C$ is a concept occurring in $\mT$ or $\mA$ as a formula or a subformula and $a$ is an individual occurring in $\mA$ 
\item the set of all formulas $\V R.C$ and $a\!:\!\V R.C$ such that $a$ is an individual occurring in $\mA$ and there exists a role $S$ such that $R \sqsubseteq_\mR S$, $S \circ S \sqsubseteq S \in \Ext(\mR)$ and $\V S.C$ is a concept occurring in $\mT$ or $\mA$ as a formula or a subformula. 
\end{itemize}

The size of $\closure(\mR,\mT,\mA)$ is polynomial in the size of $(\mR,\mT,\mA)$, where the {\em size} of a set of formulas (resp.\ a knowledge base) is the sum of the lengths of its formulas (resp.\ formulas and axioms). 

\begin{lemma} \label{lemma: tab-prop}
Procedure $\Tableau(\mR,\mT,\mA)$ runs in $2^{O(n)}$ steps and returns a rooted ``and-or'' graph $G = (V,E,\nu)$ of size $2^{O(n)}$, where $n$ is the size of $\closure(\mR,\mT,\mA)$. 
Furthermore, for every $v \in V\,$:
\begin{enumerate}
\item \label{item: HHSKA 1}
the sets $\Label(v)$, $\RFormulas(v)$ and $\DFormulas(v)$ are subsets of $\closure(\mR,\mT,\mA)$
\item \label{item: HHSKA 2}
the local tree of $v$ is a DAG (directed acyclic graph).
\end{enumerate}
\end{lemma}

\begin{proof}
The assertion \ref{item: HHSKA 1} should be clear. For the assertion \ref{item: HHSKA 2}, just observe that:  
\begin{itemize}
\item \label{item: HHSKA 2a}
if $v$ is expanded by a static rule and $w$ is a successor of $v$ then $\RFormulas(v) \subseteq \RFormulas(w)$ and $\AFormulas(v) \subseteq \AFormulas(w)$ and $\DFormulas(v) \subseteq \DFormulas(w)$ 

\item \label{item: HHSKA 2b}
if $v$ is expanded by a static rule $\rho \notin \{\rH$, $\rHP$, $\rV$, $\rConv$, $\rFormingState\}$ and $w$ is a successor of $v$ then $\RFormulas(v) \subset \RFormulas(w)$

\item \label{item: HHSKA 2c}
if $v$ is expanded by a rule $\rho \in \{\rH,\rHP,\rV,\rConv\}$ and $w$ is a successor of $v$ then $\AFormulas(v) \subset \AFormulas(w)$.
\end{itemize}
Note that, each tableau node is re-expanded at most once, by using the rule $\rConv$. It is easy to see that $G$ has size $2^{O(n)}$ and can be constructed in $2^{O(n)}$ steps. 
\myEnd
\end{proof}
} 

\begin{theorem}
\label{theorem: s-c}
Let $(\mR,\mT,\mA)$ be a knowledge base in NNF of the logic \SHI.
Then procedure $\Tableau(\mR,\mT,\mA)$\ShortVersion{ given in~\cite{nSHI-long}} runs in exponential time (in the worst case) in the size of $(\mR,\mT,\mA)$ and returns a rooted ``and-or'' graph $G = (V,E,\nu)$ such that $(\mR,\mT,\mA)$ is satisfiable iff $\Status(\nu) \neq \Unsat$. 
\myEnd
\end{theorem}

\LongVersion{The complexity issue was addressed by Lemma~\ref{lemma: tab-prop}. For the remaining assertion, see the proofs given in the next section.}
\VSpace{-1ex}


\LongVersion{
\section{Proofs of Soundness and Completeness of \CSHI}
\label{section: proofs}

\subsection{Soundness}

If $X = \{C_1,\ldots,C_n\}$ then define $\Cnj(X) = C_1 \mand \ldots \mand C_n$. If $X = \{a\!:\!C_1, \ldots, a\!:\!C_n\}$ then define $\Cnj(X) = a\!:\!(C_1 \mand \ldots \mand C_n)$. Furthermore, we define $\NegCnj(X)$ to be the NNF of $\lnot\Cnj(X)$, and define $\NegAll(\{X_1,\ldots,X_k\}) = \{\NegCnj(X_1)$, \ldots, $\NegCnj(X_k)\}$.

Let $G$ be a \CSHI-tableau for $(\mR,\mT,\mA)$. 
For each node $v$ of $G$ with $\Status(v) \in \{\Incomplete$, $\Unsat$, $\Sat\}$, let $\DSTimeStamp(v)$ be the moment at which $\Status(v)$ was changed to its final value (i.e.\ determined to be $\Incomplete$, $\Unsat$ or $\Sat$). $\DSTimeStamp$ stands for ``determined-status time-stamp''. 
For each non-state $v$ of $G$, let $\ETimeStamp(v)$ be the moment at which $v$ was expanded the last time.\footnote{Recall that, each non-state may be re-expanded at most once, using the rule $\rConv$, and that, each state is expanded at most once.}

\begin{lemma} \label{lemma: SHQWD}
Let $G = (V,E,\nu)$ be a \CSHI-tableau for $(\mR,\mT,\mA)$. For every $v \in V$:
\begin{enumerate}
\item if $\Status(v) = \Unsat$ then
  \begin{enumerate}
  \item case $\Type(v) = \State$ : $\FullLabel(v)$ is unsatisfiable w.r.t.~$\mR$ and~$\mT$
  \item case $\Type(v) = \NonState$ and $\StatePred(v) = \Null$ : $\FullLabel(v)$ is unsatisfiable w.r.t.~$\mR$ and~$\mT$

  \item case $\Type(v) = \NonState$, $v_0 = \StatePred(v) \neq \Null$ and $\SType(v_0) = \Simple$: if $v_1 = \AfterTransPred(v)$ and $\CELabel(v_1)$ is of the form $\E R.C$ then there do not exist any model $\mI$ of both $\mR$ and $\mT$ and any elements $x,y \in \Delta^\mI$ such that $(x,y) \in R^\mI$, $x \in (\FullLabel(v_0))^\mI$ and $y \in (\Label(v))^\mI$

  \item case $\Type(v) = \NonState$, $v_0 = \StatePred(v) \neq \Null$ and $\SType(v_0) = \Complex$: if $v_1 = \AfterTransPred(v)$ and $\CELabel(v_1)$ is of the form $a\!:\!\E R.C$ then there do not exist any model $\mI$ of $(\mR,\mT,\FullLabel(v_0))$ and any element $y \in \Delta^\mI$ such that $(a^\mI,y) \in R^\mI$ and $y \in (\Label(v))^\mI$
  \end{enumerate}

\item if $\Status(v) = \Incomplete$ and $\Type(v) = \State$ then
  \begin{enumerate}
  \item case $\ConvMethod(v) = 0$: $\FullLabel(v) \cup \{\NegCnj(\FmlsRC(v))\}$ is unsatisfiable w.r.t.~$\mR$ and~$\mT$
  \item case $\ConvMethod(v) = 1$: $\FullLabel(v) \cup \NegAll(\AltFmlSetsSC(v))$ is unsatisfiable w.r.t.~$\mR$ and~$\mT$
  \end{enumerate}

\item if $\Type(v) = \NonState$ and $w_1,\ldots,w_k$ are all the successors of $v$ then, for every model $\mI$ of $\mR$ and every $x \in \Delta^\mI$, 
  \begin{enumerate}
  \item case $\SType(v) = \Simple$: $x \in (\FullLabel(v))^\mI$ iff there exists $1 \leq i \leq k$ such that $x \in (\FullLabel(w_i))^\mI$
  \item case $\SType(v) = \Complex$: $\mI$ is a model of $\FullLabel(v)$ iff there exists $1 \leq i \leq k$ such that $\mI$ is a model of $\FullLabel(w_i)$.
  \end{enumerate}
\end{enumerate}
\end{lemma}

\begin{proof}
We prove this lemma by induction on both $\DSTimeStamp(v)$ and $\ETimeStamp(v)$. 

Consider the assertion~3. It should be clear for the cases when the rule expanding $v$ is not $\rConv$. So, assume that $v$ was re-expanded by the rule $\rConv$ and let $w$ be the only successor of $v$ before the re-expansion. We must have that $\Type(w) = \State$, $\Status(w) = \Incomplete$, $\DSTimeStamp(w) < \ETimeStamp(v)$, $\Label(w) = \Label(v)$, $\RFormulas(w) = \RFormulas(v)$ and $\DFormulas(w) = \DFormulas(v)$. There are the following two cases:
\begin{itemize}
\item Case $\ConvMethod(w) = 0$: By the inductive assumption~2a, $\FullLabel(w) \cup \{\NegCnj(\FmlsRC(w))\}$ is unsatisfiable w.r.t.~$\mR$ and~$\mT$. It follows that $\FullLabel(v) \cup \{\NegCnj(\FmlsRC(w))\}$ is unsatisfiable w.r.t.~$\mR$ and~$\mT$. After re-expansion $v$ has only one successor $w'$, with $\Label(w') = \Label(v) \cup \FmlsRC(w)$, $\RFormulas(w') = \RFormulas(v)$ and $\DFormulas(w') = \DFormulas(v)$. Hence, the assertion~3 holds. 

\item Case $\ConvMethod(w) = 1$: By the inductive assumption~2b, $\FullLabel(w) \cup \NegAll(\AltFmlSetsSC(w))$ is unsatisfiable w.r.t.~$\mR$ and~$\mT$. It follows that $\FullLabel(v) \cup \NegAll(\AltFmlSetsSC(w))$ is unsatisfiable w.r.t.~$\mR$ and~$\mT$. Using this, it can be observed that Steps 5-12 of procedure $\ApplyConvRule$ guarantees the assertion~3.
\end{itemize}

Consider the assertion~1. 
If $\Status(v) = \Unsat$ because $\bot \in \Label(v)$ or there exists $\{\varphi,\ovl{\varphi}\} \subseteq \Label(v)$ then $\FullLabel(v)$ is clearly unsatisfiable w.r.t.~$\mR$ and~$\mT$. So, assume that $\Label(v)$ contains neither $\bot$ nor a pair $\{\varphi,\ovl{\varphi}\}$.   
 
Consider the assertion~1a and suppose that $\Status(v) = \Unsat$ and $\Type(v) = \State$. There are three cases: $\Status(v)$ was set to $\Unsat$ either by Step~27 of procedure $\Apply$ (with $v_0 = v$) or by Step~12 of procedure $\ApplyTransRule$ (with $u = v$) or by Step~10 of procedure $\UpdateStatus$. For the first two cases, we must have that $\ConvMethod(v) = 0$ and $\FmlsRC(v) \cap \DFormulas(v) \neq \emptyset$, which implies that $\FullLabel(v)$ is unsatisfiable w.r.t.~$\mR$ and~$\mT$. The intuition behind the last inference is that $\FullLabel(v) = \AFormulas(v) \cup \NDFormulas(v)$ and $\FmlsRC(v)$ is the set of formulas which must be added to $\AFormulas(v)$. Consider the third case. Thus, $v$ has a successor $w$ with status $\Status(w) = \Unsat$, $\Type(w) = \NonState$ and $\DSTimeStamp(w) < \DSTimeStamp(v)$. The inductive assumption 1c or 1d (depending on $\SType(w)$) holds for $w$ (in the place of $v$). If $\FullLabel(v)$ is satisfied in a model $\mI$ of $\mR$ and $\mT$, then $\mI$ will violate this inductive assumption. Hence $\FullLabel(v)$ is unsatisfiable w.r.t.~$\mR$ and~$\mT$. 

Consider the assertion~1b and suppose that $\Status(v) = \Unsat$ and $\Type(v) = \NonState$ and $\StatePred(v) = \Null$. Let $w_1, \ldots, w_k$ be all the successors of $v$. It must be that, for every $1 \leq i \leq k$, $\Status(w_i) = \Unsat$, $\Type(w_i) = \NonState$, $\StatePred(w_i) = \Null$ and $\DSTimeStamp(w_i) < \DSTimeStamp(v)$. By the inductive assumption~1b, for every $1 \leq i \leq k$, $\FullLabel(w_i)$ is unsatisfiable w.r.t.~$\mR$ and~$\mT$. By the inductive assumption~3, it follows that $\FullLabel(v)$ is unsatisfiable w.r.t.~$\mR$ and~$\mT$. 

Consider the assertions~1c and~1d and suppose that $\Status(v) = \Unsat$ and $\Type(v) = \NonState$ and $\StatePred(v) \neq \Null$. There are the following cases:
\begin{itemize}
\item Case $\Status(v)$ was set to $\Unsat$ by Step~28 of procedure $\Apply$: The condition $Y \cap \DFormulas(v_0) \neq \emptyset$ of that step implies the assertions~1c and~1d. 
\item Case $\Status(v)$ was set to $\Unsat$ by Step~3 of procedure $\UpdateStatus$: Let $w_1, \ldots, w_k$ be all the successors of $v$. It must be that, for every $1 \leq i \leq k$, $\Status(w_i) = \Unsat$, $\Type(w_i) = \NonState$, $\StatePred(w_i) = \StatePred(v) \neq \Null$ and $\DSTimeStamp(w_i) < \DSTimeStamp(v)$. The inductive assumptions~1c and~1d for $w_1, \ldots, w_k$ imply the inductive hypotheses~1c and~1d for~$v$. 
\end{itemize}

The assertion~2a should be clear.

Consider the assertion~2b and suppose that $\Status(v) = \Incomplete$, $\Type(v) = \State$ and $\ConvMethod(v) = 1$. There must exist a successor $w$ of $v$ such that $\AfterTrans_\mR(w)$ holds, $\Kind(w) = \OrNode$, $\Status(w) = \Incomplete$, and $\AltFmlSetsSCP(w) = \AltFmlSetsSC(v)$. Let $w_1,\ldots,w_k$ be all the nodes in the local graph of $w$ such that, for $1 \leq i \leq k$, $\Status(w_i) = \Incomplete$ and when $\Status(w_i)$ became $\Incomplete$ a set $X_i$ of formulas was added (as an element) into $\AltFmlSetsSCP(w)$ (i.e.~$w_i$ got status $\Incomplete$ not by propagation). The setting of $\Status(w_i)$ and the addition of $X_i$ to $\AltFmlSetsSCP(w)$ occur at Steps~30 and~31 of procedure $\Apply$. We have that $\AltFmlSetsSCP(w) = \{X_1,\ldots,X_k\}$. Note that, since $\Status(w) = \Incomplete$, $k \geq 1$ and every node in the local graph of $w$ must have status $\Incomplete$ or $\Unsat$. 
There are the following two cases: 
\begin{itemize}
\item Case $\SType(v) = \Simple$: Let $\CELabel(w) = \E R.C$. For the sake of contradiction, suppose there exists a model $\mI$ of $\mR$ and $\mT$ such that $(\FullLabel(v) \cup \NegAll(\AltFmlSetsSC(v)))^\mI$ is not empty and contains an element $x$. Since $\CELabel(w) \in \Label(v)$, there exists $y \in \Delta^\mI$ such that $(x,y) \in R^\mI$ and $y \in C^\mI$. Thus, $y \in (\Label(w))^\mI$, and hence $y \in (\FullLabel(w))^\mI$ (since $\RFormulas(w) = \DFormulas(w) = \emptyset$). For every node $w'$ in the local graph of $w$ with $\Status(w') = \Unsat$, we have that $\DSTimeStamp(w') < \DSTimeStamp(v)$, and by the inductive assumption~1c, $y \notin (\Label(w'))^\mI$, and hence $y \notin (\FullLabel(w'))^\mI$. Since $y \in (\FullLabel(w))^\mI$, by the inductive assumption~3a, it follows that there exists $1 \leq i \leq k$ such that $y \in (\FullLabel(w_i))^\mI$. 
Since $X_i = \Trans_\mR(\Label(w_i),R^-)$ and $(x,y) \in R^\mI$, it follows that $x \in X_i^\mI$, which contradicts the fact that $x \in (\NegAll(\AltFmlSetsSC(v)))^\mI$. Therefore $\FullLabel(v) \cup \NegAll(\AltFmlSetsSC(v))$ must be unsatisfiable w.r.t.~$\mR$ and~$\mT$.

\item Case $\SType(v) = \Complex$: Let $\CELabel(w) = a\!:\!\E R.C$. For the sake of contradiction, suppose there exists a model $\mI$ of $\mR$, $\mT$ and $\FullLabel(v) \cup \NegAll(\AltFmlSetsSC(v))$. Since $\CELabel(w) \in \Label(v)$, there exists $y \in \Delta^\mI$ such that $(a^\mI,y) \in R^\mI$ and $y \in C^\mI$. Thus, $y \in (\Label(w))^\mI$, and hence $y \in (\FullLabel(w))^\mI$ (since $\RFormulas(w) = \DFormulas(w) = \emptyset$). For every node $w'$ in the local graph of $w$ with $\Status(w') = \Unsat$, we have that $\DSTimeStamp(w') < \DSTimeStamp(v)$, and by the inductive assumption~1d, $y \notin (\Label(w'))^\mI$, and hence $y \notin (\FullLabel(w'))^\mI$. Since $y \in (\FullLabel(w))^\mI$, by the inductive assumption~3a, it follows that there exists $1 \leq i \leq k$ such that $y \in (\FullLabel(w_i))^\mI$. Since $X_i = \Trans_\mR(\Label(w_i),R^-,a)$ and $(a^\mI,y) \in R^\mI$, it follows that $\mI$ is a model of (the ABox) $X_i$, which contradicts the fact that $\mI$ is a model of (the ABox) $\NegAll(\AltFmlSetsSC(v))$. Therefore $\FullLabel(v) \cup \NegAll(\AltFmlSetsSC(v))$ must be unsatisfiable w.r.t.~$\mR$ and~$\mT$.
\myEnd
\end{itemize}
\end{proof}

\begin{corollary}[Soundness of \CSHI]
If $G = (V,E,\nu)$ is a \CSHI-tableau for $(\mR,\mT,\mA)$ and $\Status(\nu) = \Unsat$ then $(\mR,\mT,\mA)$ is unsatisfiable.
\myEnd
\end{corollary}

This corollary follows from the assertion~1b of Lemma~\ref{lemma: SHQWD}.


\subsection{Completeness}

\begin{lemma} \label{lemma: GSDHE}
Let $G = (V,E,\nu)$ be a \CSHI-tableau for $(\mR,\mT,\mA)$. Then no node with status $\Incomplete$ is reachable from $\nu$. 
\end{lemma}

\begin{proof}
This lemma follows from the observation that, after a state $w$ getting status $\Incomplete$, all edges coming to $w$ will be deleted (see Step~1 of procedure $\ApplyConvRule$). 
\myEnd
\end{proof}

We prove completeness of \CSHI\ via model graphs. The technique has been used in \cite{Rautenberg83,Gore99,nguyen01B5SL,GoreNguyenTab07,NguyenS10FI,NguyenSzalas-CADE-22} for other logics.
A {\em model graph} (also known as a {\em Hintikka structure}) is a tuple $\langle \Delta, \mC, \mE \rangle$, where:
\VSpace{-0.5em}
\begin{itemize}
\item $\Delta$ is a finite set, which contains (amongst
    others) all individual names (occurring in the
    considered ABox)
\item $\mC$ is a function that maps each element of
    $\Delta$ to a set of concepts
\item $\mE$ is a function that maps each role to
    a binary relation on $\Delta$.
\end{itemize}

A model graph $\langle \Delta, \mC, \mE \rangle$ is {\em $\mR$-saturated} if every $x \in \Delta$ satisfies:
\begin{eqnarray}
&-& \textrm{if $C \mand D \in \mC(x)$ then $\{C,D\} \subseteq \mC(x)$}\label{eq:HGDSX 1}\\
&-& \textrm{if $C \mor D \in \mC(x)$ then $C \in \mC(x)$ or $D \in \mC(x)$}\label{eq:HGDSX 2}\\
&-& \textrm{if $\V S.C \in \mC(x)$ and $R \sqsubseteq_\mR S$ then $\V R.C \in \mC(x)$}\label{eq:HGDSX 3}\\
&-& \textrm{if $(x,y) \in \mE(R)$ then $\Trans_\mR(\mC(x),R) \subseteq \mC(y)$}\label{eq:HGDSX 4}\\
&-& \textrm{if $(x,y) \in \mE(R)$ then $\Trans_\mR(\mC(y),R^-) \subseteq \mC(x)$}\label{eq:HGDSX 5}\\
&-& \textrm{if $\E R.C \in \mC(x)$ then there exists $y \in \Delta$ s.t.\ $(x,y) \in \mE(R)$ and $C \in \mC(y)$}\label{eq:HGDSX 6}
\end{eqnarray}

A model graph $\langle \Delta, \mC, \mE \rangle$ is
{\em consistent} if no $x \in \Delta$ has $\mC(x)$ containing
$\bot$ or containing both $A$ and $\lnot A$ for some atomic
concept~$A$.

Given a~model graph $M = \langle \Delta, \mC, \mE \rangle$, the
{\em $\mR$-model corresponding to~$M$} is the interpretation
$\mI = \langle \Delta, \cdot^{\mI}\rangle$ where:
\begin{itemize}
\item $a^\mI = a$ for every individual name~$a$
\item $A^\mI = \{x \in \Delta \mid A \in \mC(x)\}$ for every concept name $A$
\item $r^\mI = \mE'(r)$ for every role name $r \in \RN$, where $\mE'(R)$ for $R \in \RN \cup \RN^{-}$ are the smallest binary relations on $\Delta$ such that:
  \begin{itemize}
  \item $\mE(R) \subseteq \mE'(R)$ 
  \item $\mE'(R^-) = (\mE'(R))^{-1}$ (i.e.,\ $(\mE'(R))^{-1} \subseteq \mE'(R^-)$ and $(\mE'(R^-))^{-1} \subseteq \mE'(R)$)
  \item if $R \sqsubseteq_\mR S$ then $\mE'(R) \subseteq \mE'(S)$
  \item if $R \circ R \sqsubseteq R \in \Ext(\mR)$ then $\mE'(R) \circ \mE'(R) \subseteq \mE'(R)$.
  \end{itemize}
\end{itemize}

Note that the smallest binary relations mentioned above always exist: for each $R \in \RN \cup \RN^{-}$, initialize $\mE'(R)$ with $\mE(R)$; then, while one of the above mentioned condition is not satisfied, extend $\mE'(R)$ for an appropriate $R \in \RN \cup \RN^{-}$ minimally to satisfy the condition.

\begin{lemma} \label{lemma: model graph}
If $\mI$ is the $\mR$-model corresponding to a consistent $\mR$-saturated model graph $\langle \Delta, \mC, \mE \rangle$, then $\mI$ is a model of $\mR$ and, for every $x \in \Delta$ and $C \in \mC(x)$, we have that $x \in C^\mI$.
\end{lemma}

\begin{proof}
Clearly, $\mI$ is a model of $\mR$. 
For the remaining assertion of the lemma, we first prove that if $(x,y) \in R^\mI$ then $\Trans_\mR(\mC(x),R) \subseteq \mC(y)$ and $\Trans_\mR(\mC(y),R^-) \subseteq \mC(x)$. 
We prove this by induction on the timestamp of the addition of the pair $(x,y)$  to $\mE'(R)$ when constructing $\mI$ from the model graph. The base case is when $(x,y) \in \mE(R)$ and follows from the assumption that $\langle \Delta, \mC, \mE \rangle$ is an $\mR$-saturated model graph.
For induction step, there are the following cases:
\begin{itemize}
\item Case $(x,y)$ is added to $\mE'(R)$ because $(y,x) \in \mE'(R^-)$: The induction hypothesis immediately follows from the inductive assumption.

\item Case $(x,y)$ is added to $\mE'(R)$ because $(x,y) \in \mE'(S)$ and $S \sqsubseteq_\mR R\,$: By~\eqref{eq:HGDSX 3}, we have that $\Trans_\mR(\mC(x),R) \subseteq \Trans_\mR(\mC(x),S)$ and $\Trans_\mR(\mC(y),R^-) \subseteq \Trans_\mR(\mC(y),S^-)$. The induction hypothesis follows from these properties and the inductive assumption with $S$ replacing~$R$.

\item Case $(x,y)$ is added to $\mE'(R)$ because $R \circ R \sqsubseteq R \in \Ext(\mR)$, $(x,z) \in \mE'(R)$ and $(z,y) \in \mE'(R)$: The induction hypothesis follows from the inductive assumption with $(x,z)$ replacing $(x,y)$ and from the inductive assumption with $(z,y)$ replacing $(x,y)$.
\end{itemize}

The remaining assertion of the lemma can then be proved by induction on the structure of $C$ in a straightforward way.
\myEnd
\end{proof}



Let $G = (V,E)$ be a \CSHI graph for $(\mR,\mT,\mA)$ and $v \in V$ be a non-state with $\Status(v) \notin \{\Unsat,\Incomplete\}$. A {\em saturation path} of $v$ is a sequence $v_0 = v$, $v_1$, \ldots, $v_k$ of nodes of $G$, with $k \geq 1$, such that $\Type(v_k) = \State$ and 
\begin{itemize}
\item for every $1 \leq i \leq k$, $\Status(v_i) \notin \{\Unsat,\Incomplete\}$
\item for every $0 \leq i < k$, $\Type(v_i) = \NonState$ and $(v_i,v_{i+1}) \in E$. 
\end{itemize}
Observe that each saturation path of $v$ is finite (by the assertion~\ref{item: HHSKA 2} of Lemma~\ref{lemma: tab-prop}). Furthermore, if $v_i$ is a non-state with $\Status(v_i) \notin \{\Unsat,\Incomplete\}$ then $v_i$ has a successor $v_{i+1}$ with $\Status(v_{i+1}) \notin \{\Unsat,\Incomplete\}$. Therefore, $v$ has at least one saturation path. 

\begin{lemma}[Completeness of \CSHI] \label{lemma: completeness}
Let $G = (V,E,\nu)$ be a \CSHI-tableau for $(\mR,\mT,\mA)$. Suppose that $\Status(\nu) \neq \Unsat$. Then $(\mR,\mT,\mA)$ is satisfiable.
\end{lemma}

\begin{proof}
By Lemma~\ref{lemma: GSDHE}, $\Status(\nu) \neq \Incomplete$. 
Hence $\nu$ has a saturation path $v_0, \ldots,v_k$ with $v_0 = \nu$. 
We construct a model graph $M = \langle \Delta, \mC, \mE\rangle$ as follows:
\begin{enumerate}
\item Let $\Delta_0$ be the set of all individuals
    occurring in $\mA$ and set $\Delta := \Delta_0$. For
    each $a \in \Delta_0$, set $\mC(a)$ to the set of all
    concepts $C$ such that $a\!:\!C \in \AFormulas(v_k)$, 
    and mark $a$ as {\em unresolved}. (Each node
    of $M$ will be marked either as unresolved or as
    resolved.) For each role $R$, set $\mE(R) :=
    \{(a,b) \mid R(a,b) \in \mA\}$.

\item While $\Delta$ contains unresolved nodes, take one
    unresolved node $x$ and do:
  \begin{enumerate}
  \item \label{step2a} For every concept $\E R.C \in
      \mC(x)$ do:
     \begin{enumerate}
     \item If $x \in \Delta_0$ then:
        \begin{itemize}
        \item Let $u = v_k$.
        \item Let $w_0$ be the node of $G$ such that $\CELabel(w_0) = x\!:\!\E R.C$. (Note that $C \in \Label(w_0)$ and $\Status(w_0) \notin \{\Unsat$, $\Incomplete\}$ since $\Status(v_k) \notin \{\Unsat$, $\Incomplete\}$.)  
        \end{itemize}

     \item Else:
        \begin{itemize}
        \item Let $u = f(x)$. ($f$ is
            a~constructed mapping that associates 
            each node of $M$ not belonging to
            $\Delta_0$ with a simple state of
            $G$. As a maintained property of
            $f$, $\Status(u) \notin \{\Unsat,\Incomplete\}$, 
	    $\E R.C \in \Label(u)$ and $\mC(x) = \AFormulas(u)$.)
        \item Let $w_0$ be the node of $G$
            such that $\CELabel(w_0) = \E R.C$. 
            (Note that $C \in \Label(w_0)$ and $\Status(w_0) \notin \{\Unsat$, $\Incomplete\}$ since $\Status(u) \notin \{\Unsat$, $\Incomplete\}$.)
        \end{itemize}

     \item Let $w_0$, \ldots, $w_h$ be a~saturation
         path of $w_0$.\\ (Note that $\Status(w_h) \notin \{\Unsat$, $\Incomplete\}$.)
     \item \label{step2(a)iv} If there does not
         exist $y \in \Delta \setminus \Delta_0$ such that $\mC(y) =
         \AFormulas(w_h)$ then: add a~new node $y$ to $\Delta$,
         set $\mC(y) = \AFormulas(w_h)$, mark $y$ as unresolved,
         and set $f(y) = w_h$. (One can consider
         $y$ as the result of sticking together the
         nodes $w_0,\ldots,w_h$ of a saturation path of $w_0$. The above
         mentioned properties of $f$ still hold.)
     \item Add the pair $(x,y)$ to $\mE(R)$.
     \end{enumerate}

  \item Mark $x$ as resolved.
  \end{enumerate}
\end{enumerate}

The above construction terminates and results in a~finite model
graph because: for every $x,x' \in \Delta \setminus \Delta_0$,
$x \neq x'$ implies $\mC(x) \neq \mC(x')$, and for every $x \in
\Delta$, $\mC(x)$ is a~subset of $\closure(\mR,\mT,\mA)$.

Note the following remarks for the remaining part of this proof: 
\begin{itemize}
\item For any node $v$ of $G$, $\RFormulas(v)$ may contain only formulas of the form $C \mand D$, $C \mor D$, $a\!:\!(C \mand D)$, or $a\!:\!(C \mor D)$. Hence, if $\varphi$ is of the form $\V R.C$, $\E R.C$, $a\!:\!\V R.C$ or $a\!:\!\E R.C$ and $\varphi \in \AFormulas(v)$ then we must have that $\varphi \in \Label(v)$.
\item After executing Step~\ref{step2(a)iv}, $\Label(w_0) \subseteq \AFormulas(w_h) = \mC(y)$. Hence, if $D \in \Label(w_0)$ then $D \in \mC(y)$. 
\end{itemize}

$M$ is a consistent model graph because $\Status(v_k) \neq \Unsat$ and if $x \in \Delta \setminus \Delta_0$ and $u = f(x)$ then $\mC(x) = \AFormulas(u)$ and $\Status(u) \neq \Unsat$. 

We show that $M$ satisfies all Conditions \eqref{eq:HGDSX 1}-\eqref{eq:HGDSX 6} of being a saturated model graph. $M$ satisfies Conditions~\eqref{eq:HGDSX 1}-\eqref{eq:HGDSX 3} because the sequence $v_0,\ldots,v_k$ is a saturation path of $v_0$, and at Step~\ref{step2a}, the sequence $w_0,\ldots,w_h$ is a saturation path of~$w_0$. $M$ satisfies Condition~\eqref{eq:HGDSX 6} because at Step~\ref{step2a}, $C \in \Label(w_0)$, and hence $C \in \mC(y)$.

Consider Condition~\eqref{eq:HGDSX 4}:
\VSpace{-1ex}
\begin{itemize}
\item Assume $x \in \Delta$ and $(x,y) \in \mE(R)$. We show that $\Trans(\mC(x),R) \subseteq \mC(y)$.

\item 
Consider the case $x \in \Delta_0$ and $\V R.D \in \mC(x)$. 
We show that $D \in \mC(y)$. 
Since $\V R.D \in \mC(x)$, we have that $x\!:\!\V R.D \in \Label(v_k)$. If $y \in \Delta_0$, then $R(x,y) \in \mA$ (since $(x,y) \in \mE(R)$), and hence $y\!:\!D \in \AFormulas(v_k)$ (due to the tableau rule $\rV$), and hence $D \in \mC(y)$.
Assume that $y \notin \Delta_0$ and $y$ is created at Step~\ref{step2(a)iv}. Since $(x\!:\!\V R.D)$ belongs to the label of $u = v_k$, by the tableau rule $\rTransP$, $D \in \Label(w_0)$, and hence $D \in \mC(y)$.

\item 
Consider the case $x \in \Delta_0$, $\V S.D \in \mC(x)$, $R \sqsubseteq_\mR S$ and $S \circ S \sqsubseteq S \in \Ext(\mR)$. 
We show that $\V S.D \in \mC(y)$. 
Since $\V S.D \in \mC(x)$, we have that $x\!:\!\V S.D \in \Label(v_k)$. If $y \in \Delta_0$, then $R(x,y) \in \mA$ (since $(x,y) \in \mE(R)$), and hence $y\!:\!\V S.D \in \AFormulas(v_k)$ (due to the tableau rule $\rV$), and hence $\V S.D \in \mC(y)$.
Assume that $y \notin \Delta_0$ and $y$ is created at Step~\ref{step2(a)iv}. Since $(x\!:\!\V S.D)$ belongs to the label of $u = v_k$, by the tableau rule $\rTransP$, $\V S.D \in \Label(w_0)$, and hence $\V S.D \in \mC(y)$.

\item Consider the case $x \notin \Delta_0$ and Step~\ref{step2a} at which the pair $(x,y)$ is added to $\mE(R)$. 
  \begin{itemize}
  \item Suppose that $\V R.D \in \mC(x)$. We show that $D \in \mC(y)$. 
Since $\V R.D \in \mC(x)$ and $\mC(x) = \AFormulas(u)$, we have that $\V R.D \in \Label(u)$. By the tableau rule $\rTrans$, it follows that $D \in \Label(w_0)$, and hence $D \in \mC(y)$.
  \item Suppose that $\V S.D \in \mC(x)$, $R \sqsubseteq_\mR S$ and $S \circ S \sqsubseteq S \in \Ext(\mR)$. We show that $\V S.D \in \mC(y)$. 
Since $\V S.D \in \mC(x)$ and $\mC(x) = \AFormulas(u)$, we have that $\V S.D \in \Label(u)$. By the tableau rule $\rTrans$, it follows that $\V S.D \in \Label(w_0)$, and hence $\V S.D \in \mC(y)$.
  \end{itemize}
\end{itemize}

\VSpace{-1ex}
Consider Condition~\eqref{eq:HGDSX 5}:
\VSpace{-1ex}
\begin{itemize}
\item Assume $y \in \Delta$ and $(x,y) \in \mE(R)$. We show that $\Trans(\mC(y),R^-) \in \mC(x)$.

\item Consider the case $y \in \Delta_0$. We must have that $x \in \Delta_0$ and $R(x,y) \in \mA$. 
  \begin{itemize}
  \item If $\V R^-.D \in \mC(y)$ then $y\!:\!\V R^-.D \in \Label(v_k)$, which implies that $x\!:\!D \in \AFormulas(v_k)$ (by the tableau rule $\rV$), and hence $D \in \mC(x)$.
  \item If $\V S.D \in \mC(y)$, $R^- \sqsubseteq_\mR S$ and $S \circ S \sqsubseteq S \in \Ext(\mR)$ then $y\!:\!\V S.D \in \Label(v_k)$, which implies that $x\!:\!\V S.D \in \AFormulas(v_k)$ (by the tableau rule $\rV$), and hence $\V S.D \in \mC(x)$.
  \end{itemize}

\item Consider the case $y \notin \Delta_0$ and Step~\ref{step2a} at which the pair $(x,y)$ is added to $\mE(R)$. 
  \begin{itemize}
  \item Suppose that $\V R^-.D \in \mC(y)$. We show that $D \in \mC(x)$.
Since $\V R^-.D \in \mC(y)$, we have that $\V R^-.D \in \Label(w_h)$. Since $\Status(w_h) \notin \{\Unsat$, $\Incomplete\}$, we must have that:
     \begin{itemize}
     \item if $u$ is a simple node then $x \notin \Delta_0$ and $D \in \AFormulas(u)$, and hence $D \in \mC(x)$
     \item else $u = v_k$, $x \in \Delta_0$ and $x\!:\!D \in \AFormulas(u)$, and hence $D \in \mC(x)$.
     \end{itemize}

  \item Suppose that $\V S.D \in \mC(y)$, $R^- \sqsubseteq_\mR S$ and $S \circ S \sqsubseteq S \in \Ext(\mR)$. We show that $\V S.D \in \mC(x)$.
Since $\V S.D \in \mC(y)$, we have that $\V S.D \in \Label(w_h)$. Since $\Status(w_h) \notin \{\Unsat$, $\Incomplete\}$, we must have that:
     \begin{itemize}
     \item if $u$ is a simple node then $x \notin \Delta_0$ and $\V S.D \in \AFormulas(u)$, and hence $\V S.D \in \mC(x)$
     \item else $u = v_k$, $x \in \Delta_0$ and $x\!:\!\V S.D \in \AFormulas(u)$, and hence $\V S.D \in \mC(x)$.
     \end{itemize}
  \end{itemize}

\end{itemize}

Therefore $M$ is a consistent saturated model graph.

By the definition of \CSHI graphs for $(\mR,\mT,\mA)$ and the
construction of $M$: if $(a\!:\!C) \in \mA$ then $C \in
\mC(a)$; if $R(a,b) \in \mA$ then $(a,b) \in \mE(R)$; and $\mT
\subseteq \mC(a)$ for all $a \in \Delta_0$. We also have that
$\mT \subseteq \mC(x)$ for all $x \in \Delta \setminus
\Delta_0$. Hence, by Lemma~\ref{lemma: model graph}, the
interpretation corresponding to $M$ is a~model of $(\mR,\mT,\mA)$.
\myEnd
\end{proof}
} 


\section{Conclusions}
\label{section: conc}

We have given the first cut-free \EXPTIME (optimal) tableau decision procedure for checking satisfiability of a knowledge base in the description logic \SHI. 
Our decision procedure is novel: in contrast to~\cite{GoreNguyenTab07,GoreW09,GoreW10}, it deals also with ABoxes; in contrast to~\cite{GoreNguyenTab07,NguyenSzalas-SL}, it does not use cuts; in contrast to~\cite{NguyenS10TCCI,NguyenS10FI}, it deals also with inverse roles; and in contrast to~\cite{Nguyen-ALCI}, it deals also with transitive roles and hierarchies of roles. 
The procedure can be implemented with various optimizations as in~\cite{Nguyen08CSP-FI} to give an efficient complexity-optimal program for checking satisfiability of a knowledge base in the popular DL \SHI. 
\LongVersion{We intend to extend our methods for other DLs.}




\end{document}